\documentclass[lettersize,draftcls,onecolumn]{IEEEtran}
\usepackage{amsmath,amsfonts, amssymb, theorem}
\usepackage{algorithmic}
\usepackage{algorithm}
\usepackage{array}
\usepackage[caption=false,font=normalsize,labelfont=sf,textfont=sf]{subfig}
\usepackage{textcomp}
\usepackage{stfloats}
\usepackage{url}
\usepackage{verbatim}
\usepackage{graphicx}
\usepackage{cite}
\usepackage{hyperref}
\usepackage{xcolor}
\def\C{{\bf C}}

\def\EE{{\bf E}}

\def\X{{\bf X}}
\def\x{{\bf x}}

\def\G{{\bf G}}
\def\U{{\bf U}}
\def\H{{\bf H}}

\def\u{{\bf u}}

\def\I{{\bf I}}
\def\P{{\bf P}}
\def\R{{\bf R}}

\def\W{{\bf W}}
\def\F{{\bf F}}
\def\M{{\bf M}}
\def\Y{{\bf Y}}
\def\H{{\bf H}}
\def\X{{\bf X}}
\def\A{{\bf A}}
\def\B{{\bf B}}
\def\D{{\bf D}}
\def\Z{{\bf Z}}
\def\W{{\bf W}}
\def\Y{{\bf Y}}
\def\GG{{\bf \Gamma}}
\def\tr{\operatorname{tr}}
\def\det{\operatorname{det}}

\def\det{\operatorname{det}}

\def\argmin{\operatorname{argmin}}

\def\He{\textnormal{H}}
\def\Gras{\mathbb{G}(M,{\mathbb{C}}^{T})}

\def\St{\mathbb{S}_t(M,{\mathbb{C}}^{T})}

{\theorembodyfont{\slshape}
\newtheorem{theorem}{Theorem}
\newtheorem{proposition}{Proposition}
\newtheorem{corollary}{Corollary}
\newtheorem{lemma}{Lemma}
}
{\theorembodyfont{\rmfamily}
\newtheorem{remark}{Remark}
\newtheorem{definition}{Definition}

}
\newenvironment{proof}{\textbf{Proof:}}{\hfill$\square$}

\begin{document}

\title{Constrained Riemannian Noncoherent Constellations for the MIMO Multiple Access Channel}

\author{Javier~{\'A}lvarez-Vizoso,
        Diego~Cuevas, 
        Carlos~Beltr{\'a}n,
        Ignacio~Santamaria, 
        V{\' i}t~Tu{\v c}ek
        and Gunnar~Peters%
\thanks{J. {\'A}lvarez Vizoso, D. Cuevas, and I.Santamaria are with the Department of Communications Engineering, Universidad de Cantabria, 39005 Santander, Spain (e-mail: javier.alvarezvizoso@unican.es; diego.cuevas@unican.es; i.santamaria@unican.es).}
\thanks{C. Beltr{\'a}n is with the Department of Mathematics, Statistics and Computing, Universidad de Cantabria, 39005 Santander, Spain (e-mail: carlos.beltran@unican.es).}
\thanks{V. Tucek and G. Peters are with the Department of Wireless Algorithms, Huawei Technologies, 16440 Kista, Sweden (email: vit.tucek@huawei.com; gunnar.peters@huawei.com)}}

\markboth{IEEE Transactions on Information Theory,~Vol.~x, No.~x, August~2022}%
{Shell \MakeLowercase{\textit{et al.}}: A Sample Article Using IEEEtran.cls for IEEE Journals}

\IEEEpubid{0000--0000/00\$00.00~\copyright~2022 IEEE}

\maketitle

\begin{abstract}
We consider the design of multiuser constellations for a multiple access channel (MAC) with $K$ users, with $M$ antennas each, that transmit simultaneously to a receiver equipped with $N$ antennas through a Rayleigh block-fading channel, when no channel state information (CSI) is available to either the transmitter or the receiver. In full-diversity scenarios where the coherence time is at least $T\geq (K+1)M$, the proposed constellation design criterion is based on the asymptotic expression of the multiuser pairwise error probability (PEP) derived by Brehler and Varanasi in \cite{Brehler01noco}. In non-full diversity scenarios, for which the previous PEP expression is no longer valid, the proposed design criteria is based on proxies of the PEP recently proposed by Ngo and Yang in \cite{ngoyang}. Although both the PEP expression and its bounds or proxies were previously considered intractable for optimization, in this work we derive their respective unconstrained gradients. These gradients are in turn used in the optimization of the proposed cost functions in different Riemannian manifolds representing different power constraints. In particular, in addition to the standard unitary space-time modulation (USTM) leading to optimization on the Grassmann manifold, we consider a more relaxed per-codeword power constraint leading to optimization on the so-called {\it oblique manifold}, and an average power constraint leading to optimization on the so-called {\it trace manifold}. Equipped with these theoretical tools, we design multiuser constellations for the MIMO MAC in full-diversity and non-full-diversity scenarios with state-of-the-art performance in terms of symbol error rate (SER).

\end{abstract}


\section{Introduction}

\IEEEPARstart{I}{n} multiple-input multiple-output (MIMO) noncoherent wireless communications over fast fading channels, the channel state information (CSI) is assumed to be unknown at both the transmitter and receiver. It is usual to consider in the study of noncoherent communications a block-fading model in which the MIMO channel matrix with $M$ transmit and $N$ receive antennas remains constant during a $T$-symbol coherence interval, after which it changes to a new independent realization for another $T$ symbols. In the single-user case and under additive Gaussian noise, it was proved by Hochwald and Marzetta \cite{marzetta_capacity, Hochwald00} that the $T \times M$ space-time transmit matrices that achieve the ergodic noncoherent capacity for the MIMO block-fading model can be factored as the product of an isotropically distributed $T \times M$ truncated unitary matrix, also called Stiefel matrix, and a diagonal $M \times M$ matrix with real nonnegative entries. Further, when $T>>M$ the nonzero entries of the diagonal matrix take the same value, showing that in this regime it is optimal to transmit unitary space-time codewords $\X^H\X = \I_M$. Using the same signal model, Zheng and Tse \cite{tse_noncoherent} proved that at high signal-to-noise ratio (SNR) and when $T \geq 2M$, ergodic capacity can be achieved by transmitting isotropically distributed unitary matrices. Motivated by these information-theoretic results, numerous methods for the design of single-user constellations formed by truncated unitary signal matrices, called unitary space-time modulations (USTM), have been investigated and proposed over the last decades \cite{Hochwald_TIT2000,HochwaldDiff00,Hughes_TIT2000,Beko07,gohary,cipriano,Zhao04,HanTIT06,cube_split,Cuevas21WSA,AlvarezEusipco22}. In MIMO noncoherent constellations, information is carried by the column span (i.e., a subspace) of the transmitted $T \times M $ matrix, ${\bf X}$. The problem of designing single-user noncoherent codebooks is thus closely related to finding optimal packings in Grassmann manifolds \cite{tse_noncoherent,Conway96}, and the resulting constellations are referred to as Grassmannian constellations.

In the multiuser case, the design of noncoherent constellations is significantly more complex, as many of the theoretical results that exist for the single-user case (as well as the insights gained from them), such as the optimality of unitary space-time or Grassmannian constellations at high-SNR, are no longer true. In this work, we consider the design of noncoherent constellations for the MIMO multiple access channel (MAC), a problem for which there is no satisfactory solution yet. In the MAC several users transmit information simultaneously over the same bandwidth and at the same channel use or time slot to a common receiver. A common example is the uplink channel in broadband cellular communications, where several users communicate with a base station (BS). In the case of coherent communications with perfect channel state information (CSI) at the receiver or BS, capacity results for the MIMO MAC can be found in \cite{Goldsmith_MIMOJSAC}. For instance, it is well-known that for the 2-user MAC the capacity region is a pentagon, and the Pareto optimal achievable rate pairs $(R_1,R_2)$ at the corner points of the pentagon are reached by successive cancellation.

For noncoherent communications, however, the full capacity region of the MIMO-MAC is unknown. For the $K$-user single-input multiple-output (SIMO) MAC, it was conjectured by Shamai and Marzetta in \cite{Shamai02} that for block-fading channels with coherence time $T>1$ the sum capacity can be achieved by no more than $K=T$ users, which is supported by asymptotic analysis and simulation results. For the two-user MIMO MAC an achievable DoF (degrees of freedom) region has been proposed in \cite{Utkovski13}. The optimal DoF region for a two-user SIMO MAC has been derived in \cite{Ngo18}. Existing theoretical results however do not provide clear insights regarding the structure of the transmit space-time matrices for the MIMO MAC.

For the SIMO case with single-antenna users, energy-based noncoherent constellation designs have been proposed for the uplink channel in \cite{Goldsmith16a, Goldsmith16b, Vucetic19}, and noncoherent schemes based on a differential phase-shift keying (DPSK) modulation have been recently proposed in \cite{Baeza1, Baeza2}. These energy-based or DPSK-based designs, however, cannot be directly extended to the MIMO case. Of particular importance for the $K$-user MIMO MAC is the work of Brehler and Varanasi in \cite{Brehler01noco}, where the authors derived an asymptotic expression of the joint pairwise probability of error (PEP) of the optimum receiver and showed that, to ensure full diversity of $NM$ for each user, the coherence time must be at least $T \geq (K+1)M$. However, the PEP expression in \cite{Brehler01noco} was considered to be intractable for optimization, so none of the subsequent studies have used it as a criterion to design multiuser constellations. Most of the proposed criteria in the literature either optimize single-user Grassmannian designs with or without partitioning; that is, using independently designed single-user codebooks, or designing a large single-user codebook that is then partitioned according to some subspace distance measure into $K$ smaller single-user codebooks \cite{Varanasi01,Brehler01noco,joint_ngo,joint_guillaud}. 

As an alternative to the PEP criterion, in \cite{ngoyang} Ngo and Yang recently proposed two PEP proxies that have a very natural geometric meaning in terms of separating joint detection hypothesis. However, the proposed proxies are functions of the eigenvalues of a certain matrix and therefore their optimization is considered challenging in \cite{ngoyang}. It is interesting to note at this point that while the exact asymptotic PEP expression in \cite{Brehler01noco} is only valid in uplink channels where the full-diversity condition $T \geq (K+1)M$ is met, which we will refer to from now on as full-diversity scenarios, the proxies proposed in  \cite{ngoyang} are valid in non-full diversity scenarios where  $T < (K+1)M$. In fact, as we will show in this work, the PEP cost function of Brehler and Varanasi and the PEP proxies of Ngo and Yang, yield two complementary designs that can be applied, respectively, to full-diversity and the non-full diversity scenarios. In addition, as also shown in this work, both cost functions can be optimized on different manifolds representing the different power constraints typically employed in the $K$-user MIMO MAC.


The main contributions of the paper are the following:
\begin{enumerate}

\item We have developed Riemannian optimization techniques for designing multiuser noncoherent codebooks for the MIMO MAC in manifolds other than the complex Grassmannian. These manifolds correspond to alternative power normalizations to the one used in unitary space-time modulations (USTM), which need not be optimal for noncoherent multiuser communications. In particular, in addition to the standard Grassmann manifold, we have considered the complex oblique manifold and the trace manifold, another type of oblique manifold, (Sec. \ref{sec:manifolds}), resulting from a per-codeword power constraint and an average power constraint, respectively. 


\item We have obtained, for the first time in the literature, closed-form formulas for the gradients of cost functions previously proposed for the design of multiuser noncoherent constellations for the MAC, but whose optimization was so far considered to be intractable. These functions are union bounds, i.e., sums over all joint codewords of the dominant factor (when one user is in error) of the asymptotic PEP derived by Brehler and Varanasi \cite{Brehler01noco}, and the $\beta$ and $\delta$ functions (bounds of the PEP) proposed by Ngo and Yang \cite{ngoyang}.

\item Using these gradient expressions, we have developed Riemannian techniques over different manifolds to optimize the $\beta$ and $\delta$ functions proposed in \cite{joint_ngo} for non-full diversity scenarios, as well as to optimize the exact PEP asymptotic expression derived by Brehler and Varanasi \cite{Brehler01noco} for full-diversity scenarios. 

\item Our results show that the best performing designs in the non-full diversity case, i.e. $T < (K+1)M$, are those obtained using the $\delta$ cost function on the trace manifold. Whereas in the full-diversity case, i.e. $T\geq (K+1)M$, the best performing designs are those obtained with the union bound of the asymptotic PEP on the Grassmann manifold.
\end{enumerate}

The remainder of this paper is organized as follows. Sec. \ref{sec:system_model} introduces the system model for the $K$-user MIMO MAC and the optimal multiuser maximum likelihood detector. Sec. \ref{sec:manifolds} describes the different Riemannian manifolds considered for codebook optimization, along with their projection and retraction steps. In Sec. \ref{subsec:asym_pep} the asymptotic joint PEP is reviewed in order to introduce the optimization cost function for full-diversity designs and its Riemannian gradient is computed in Sec. \ref{subsec:pep_grad}. In Sec. \ref{subsec:proxies} the optimization cost functions for non-full diversity designs are presented, and their gradients are obtained in Sec. \ref{subsec:proxy_grad}. Sec. \ref{subsec:resultsProxies} presents and analyzes the simulation results for the noncoherent multiuser constellation designs in non-full diversity scenarios, while Sec. \ref{subsec:resultsUB} discusses the results corresponding to the designs in full-diversity scenarios. Finally, we present our conclusions in Sec. \ref{sec:conclusions}. The paper also includes an Appendix \ref{app:preliminaries} for mathematical background on Riemannian manifolds.


 \textit{Notation}: In this paper, matrices are denoted by bold-faced upper case letters, column vectors are denoted by bold-faced lower case letters, and scalars are denoted by light-faced lower case letters. The superscripts $(\cdot)^T$ and $(\cdot)^H$ denote transpose and Hermitian conjugate, respectively. The trace and determinant of a matrix $\A$ will be denoted, respectively, as $\tr(\A)$ and $\det(\A)$. We denote by ${\rm diag} ({\bf a})$ a diagonal matrix whose diagonal is ${\bf a}$, and $\I_n$ denotes the identity matrix of size $n$. ${\cal CN}(0,1)$ denotes a complex proper Gaussian distribution with zero mean and unit variance, $\x \sim {\cal CN}_{n}({\bf 0}, \R)$ denotes a complex Gaussian vector in $\mathbb{C}^n$ with zero mean and covariance matrix $\R$. 
 $\Gras$ denotes the complex Grassmann manifold of $M$-dimensional subspaces of the $T$-dimensional complex vector space $\mathbb{C}^T$. $\mathbb{S}t(M,\mathbb{C}^T)$ denotes the complex Stiefel manifold of unitary $M$-frames in $\mathbb{C}^T$. Unless stated otherwise $\log$ refers to the natural logarithm. Some background material about the Stiefel and Grassmann manifolds, which is needed for the paper, is relegated to the Appendix. Additional notation is introduced as needed in the text. 
\section{Manifold Optimization for the MIMO MAC}\label{sec:mu_detectors}

\subsection{System model}\label{sec:system_model}

We consider a noncoherent MIMO MAC with $K$ transmitters, or users, simultaneously transmitting to a common receiver, or base station. To keep notation simple, we assume that all users have the same number of transmit antennas $M$, (the extension to a different number of antennas per user is straightforward), and the receiver has $N$ antennas. The channel of user $k$ is $\H_k \in \mathbb{C}^{M\times N}$ has a Rayleigh fading distribution ($\H_k(i,j) \sim {\cal CN} (0,1)$) and is assumed to remain constant over $T$ symbol periods, over which communication occurs. In the next transmission block the channels of all users change to an independent realization (block-fading channel). User $k$ transmit at rate $R_k$ (bits/channel use), so within a coherence block sends a matrix chosen equiprobably from a codebook ${\cal C}_k=\{\X_{k,1},\ldots, \X_{k,L_k} \} $ with $L_k = 2^{R_kT}$. Unlike the single-user case, for the MAC the transmitted matrices do not have to be necessarily semi-unitary or Stiefel ($\X_{k,i}^{\He} \X_{k,i} = \I_M $). Obviously there is a total power constraint
\begin{equation}
    \frac{1}{L_k} \sum_{i=1}^{L_k} \tr \left(\X_{k,i}^{\He}\X_{k,i}\right) = M, \quad \forall k.
    \label{eq:power constraint}
\end{equation}   

Let us consider for notational simplicity the two-user MIMO MAC. We adhere to the notation in \cite{Brehler01co,Brehler01noco} and define the $T \times 2M $ matrix of transmitted codewords $\F_i = [\X_{1,i_1}, \X_{2,i_2}]$. Note that even if $\X_{k,i_1} \in \Gras$ for $k=1,2$, the multiuser codeword $\F_i$ is not a Stiefel matrix anymore ($\F_i^{\He} \F_i \neq \I_{2M}$). The set of multiuser codewords is
 \[
 {\cal{F}} = \{ \F_i = [\X_{1,i_1}, \X_{2,i_2}], \,\, \X_{1,i_1}  \in {\cal{C}}_1, \X_{2,i_1}  \in {\cal{C}}_2 \},  
 \]
and has cardinality $|{\cal{F}}| = |{\cal{C}}_1 \times {\cal{C}}_2|= L_1 L_2$.

Each user can have a different SNR due to the different path loss. Let $\beta_k \rho$ be the SNR for user $k$, where $\rho$ is a reference SNR that, without loss of generality, we will take as the SNR of user 1 so $\beta_1 = 1$. The model generalizes to $K$ users. When the multiuser codeword $\F_i = [\X_{1,i_1}, \X_{2,i_2}]$ is transmitted, the signal received at the BS is
\begin{equation}
\Y = \X_{1,i_1} \H_1 + \sqrt{\beta_2} \X_{2,i_2} \H_2 + \sqrt{\frac{M}{T \rho}}\W.
\label{eq:Y}
\end{equation}

The conditional distribution of the observations $\Y$ given the transmitted multiuser codeword $\F_i$ is
\begin{equation}
 p(\Y|\F_i) = \frac{1}{\pi^{TN} \det(\R_i)^N} \exp\tr \left( - \R_{i}^{-1} \Y\Y^{\He} \right),
 \label{eq:pdf}
 \end{equation}
so each column of $\Y$ follows a zero-mean complex normal distribution with covariance matrix $\R_i = \X_{1,i_1}\X_{1,i_1}^{\He} + \beta_2 \X_{2,i_2}\X_{2,i_2}^{\He} + \frac{M}{T \rho} \I_T$. If the codewords are chosen with equal probability, the optimum Maximum Likelihood (ML) detector is
\begin{equation}
 \hat{\F}_{i} = {\rm arg}\min_{\F_i  \in {{\cal{F}}}} \, \tr \left(\Y^{\He} \R_{i}^{-1} \Y \right) + N \log \det(\R_{i}).
 \label{eq:MLdet}
 \end{equation}

Notice that the ML detector at the BS needs to know the SNR of all users. The SNR depends primarily on the path loss and therefore varies on a much slower temporal scale than the multipath fading. It is therefore feasible to have this long-term CSI available at the BS.

\subsection{Riemannian manifolds for noncoherent multiuser constellation designs}\label{sec:manifolds}

For a review of the manifold geometry needed in our optimization methods we refer the reader to Appendix \ref{app:preliminaries}. 

Under the usual USTM assumption (Grassmannian constellations) used in most previous works, the codewords transmitted by users are normalized as:
\begin{equation}
\X_{k,i}^{\He}\X_{k,i} = \I_M,\text{ so that } ||\X_{k,i}||_F^2 = M,\; \forall\X_{k,i}\in\mathcal{C}_k.
\label{eq:normUSTM}
\end{equation}
Under this constraint the codewords are represented by Stiefel matrices up to unitary transformations, so that optimization must be performed on the Grassmannian manifold $\Gras$. The constraint \eqref{eq:normUSTM} means that the signals transmitted by the different antennas are orthogonal to each other, all of them with unit power.

A more relaxed constraint is to require the total per-codeword transmit power to be normalized but without requiring that the signals transmitted by different antennas be orthogonal; that is, without requiring each codeword to be a Stiefel matrix. Let us recall that the use of USTM is not necessarily optimal in the MIMO MAC. That is to say, one requires only that
\begin{equation}
||\X_{k,i}||_F^2 = M,\; \forall\X_{k,i}\in\mathcal{C}_k.
\label{eq:normoblique}
\end{equation}
This realizes the codewords as points in the complex sphere of radius $\sqrt{M}$: take the $T\times M$ matrix $\X_{k,i}$ and flatten it to a vector of length $TM$ of Euclidean norm $1$ normalizing it by $1/\sqrt{M}$. Then the users' constellations $\mathcal{C}_1\times\cdots\times\mathcal{C}_K$ correspond to a point in the so-called {\it oblique manifold}, denoted as $\mathcal{OB}_{\mathbb{C}}(TM,\sum_{k=1}^K L_k)$, which is the product of as many complex spheres as codewords $\mathbb{S}^{TM-1}\times\cdots\times\mathbb{S}^{TM-1}$. Indeed, every fixed-norm codeword corresponds to a point in a sphere of dimension $TM$, and the number of spheres needed is the added cardinality of all constellations $\sum_k^K L_k$. To optimize constellation points on this oblique manifold one just needs to project unconstrained gradients onto spheres, and to do the retraction, unflatten the vectors to restore $T\times M$ matrices, and renormalize each point from $1$ to $\sqrt{M}$.

Finally, the least stringent power constraint normalizes the average transmit power of each user $k$:
\begin{equation}
    \frac{1}{L_k} \sum_{i=1}^{L_k} \tr \left(\X_{k,i}^{\He}\X_{k,i}\right) = M,\quad k=1,\dots,K.
    \label{eq:normaverage}
\end{equation} 
In this case, it is the whole constellation of every user that behaves as a point on a sphere of radius $\sqrt{L_k M}$, that is to say, there is a vector that represents the constellation by flattening the concatenated matrices of the $L_k$ codewords of a user. Simplifying this case to the situation where all the users have the same number of codewords, $L_k=L$, the multiuser constellation $\mathcal{C}_1\times\dots\times\mathcal{C}_K$ corresponds to points in a different oblique manifold that we call the {\it trace manifold} denoted as $\mathbb{T}r(K,L,M,\mathbb{C}^T) = \mathcal{OB}_{\mathbb{C}}(TML,K)$, which corresponds to the product of as many complex spheres as users.


To optimize a function $f$ on a general manifold $\mathcal{M}\subset\mathbb{C}^T$ we just need to compute the unconstrained complex Euclidean gradient $Df$ and use the corresponding projector $\textsc{P}_{\X}:T_{\X}\mathbb{C}^T\rightarrow T_{\X}\mathcal{M}$, along with a reasonable retraction function $\textsc{R}_{\X}:T_{\X}\mathcal{M}\rightarrow\mathcal{M}$. We summarize here these steps for the manifolds investigated in our present work corresponding to the different power constraints discussed above (we use the notation of standard references such as \cite{absil}).

\begin{itemize}
    \item \textbf{Grassmannian manifold} $\Gras$: a Stiefel (or semi-unitary) matrix per user's codeword.
    \begin{itemize}
        \item \emph{Projection}: $\textsc{P}_{\X}(\dot\Z) = (\I_T-\X\X^{\He})\dot{\Z}$.
        \item \emph{Retraction}: $\textsc{R}_{\X}=\text{QR}$ decomposition.
    \end{itemize}
    \item \textbf{Oblique manifold} $\mathcal{OB}_{\mathbb{C}}(TM,\sum_i L_i)$ (per-codeword power constraint): a complex sphere of radius $\sqrt{M}$ per user's codeword.
    \begin{itemize}
        \item \emph{Projection}: $\textsc{P}_{\X}(\dot\Z) = \dot{\Z} - \X\cdot\text{ddiag}(\X^{\He}\dot{\Z}) $, where $\text{ddiag}(\W)$ is the diagonal matrix whose diagonal is that of $\W$.
        \item \emph{Retraction}: $\textsc{R}_{\X}(\dot\Z)= (\X + \dot{\Z})[\text{ddiag}(\X+\dot{\Z})^{\He}(\X+\dot{\Z})]^{-1/2}$ and normalization by scaling from $1$ to $\sqrt{M}$.
    \end{itemize}
    \item \textbf{Trace manifold} $\mathbb{T}r(K,L,M,\mathbb{C}^T)=\mathcal{OB}_{\mathbb{C}}(TML,K)$ (average per-user power constraint): a different oblique manifold with one complex sphere per user codebook.
    \begin{itemize}
        \item \emph{Projection}: $\textsc{P}_{\X}(\dot\Z) = \dot{\Z} - \X\cdot\text{ddiag}(\X^{\He}\dot{\Z}) $. A different realization of the constraint that we also employ as it proves to be numerically useful, denoting $\X_C:=\sum_{i=1}^{L_k}\X_{k,i}$, is
        $$
        \textsc{P}_{\X}(\dot{\Z}) = \dot{\Z} -\frac{\Re\,\langle \dot{\Z},\,\X_C\rangle}{\langle\,\X_C ,\,\X_C\,\rangle}\X_C.
        $$
        \item \emph{Retraction}: $\textsc{R}_{\X}(\dot\Z)= (\X + \dot{\Z})[\text{ddiag}(\X+\dot{\Z})^{\He}(\X+\dot{\Z})]^{-1/2}$ and renormalization by scaling from $1$ to
        $$
        \X_{k,i} \mapsto \sqrt{\frac{ML_k}{\sum_{i=1}^{L_k} ||\X_{k,i}||^2_F}}\X_{k,i}.
        $$
    \end{itemize}
\end{itemize}

Finally, let us introduce some notation conventions in order to compute gradients or partial derivatives on manifolds in the multiuser setting.
In particular, notice that a joint codeword, $\F_i = [\X_{1,i_1},\dots,\X_{K,i_K}]$, in the multiuser constellation $\mathcal{C}$ consists of a choice of a single-user codeword $\X_{k,i_k}\in\mathcal{C}_k$ from each of the users' packings or constellations. This entails that a multiuser function $f(\mathcal{C})$ may depend on each of the individual codewords $\X_{k,i_k}$ through different pairs (transmitted$\rightarrow$received after ML detection) of multiuser codewords $(\F_i\rightarrow\F_j)$. The transmitted symbol $\X_{k,i_k}$ can be part of the columns of either $\F_i$ or $\F_j$, or both. Therefore, we will need to extract partial derivatives of functions of joint multiuser codewords with respect to single-user codewords, i.e. with respect to only a particular subset of $M$ columns from the total $KM$ columns of a joint codeword.

Let us consider any component $x_{ab}$ of any codeword $\X\in\mathcal{C}_k$ from the $k$-th user constellation as the varying parameter, so that the movement in $\X$ is given by $\X +t\dot{\Z}$, where the matrix $\dot\Z$ will usually be one of the $\dot{\Z}_{ab}$, which is $1$ at row $a$ and column $b$ and zero elsewhere. Then the joint constellation $\mathcal{C}$ changes to $\mathcal{C}(t)$ by updating any multiuser (or joint) codeword that includes $\X$, i.e. changing $\F=[\X_1\dots\X\dots\X_K]$ into $\F(t)=[\X_1\dots\X +t\dot{\Z}\dots\X_K]$. Correspondingly, the value of any function that depends on the complete joint constellation $f(\mathcal{C})$ changes to $f(\mathcal{C}(t))$. We shall specify with respect to which single-user codeword the constellation is varying by writing $f(\mathcal{C}(t)\vert\X+t\dot{\Z})$. If a single-user codeword is present in a joint codeword we shall write $\X\in\F$, and denote by $f(\F_i(t)\vert\X+t\dot{\Z},\F_j)$ the updated value of a function of several multiuser codewords, e.g. $f(\F_i,\F_j)$.

Recalling the relationship of a matrix derivative and its gradient included in the appendix, Eq. (\ref{eq:gradient_def}), we may write the directional derivative in the $\dot\Z$ direction of any function of a joint constellation, with respect to any single user codeword $\X$, as:
\begin{equation}\label{eq:jointderivative}
D_{\X}f(\mathcal{C})(\dot\Z) = \frac{d}{dt}\bigg\vert_{t=0}f(\mathcal{C}(t)\vert\X+t\dot{\Z}) = \Re(\langle\;D_{\X} f(\mathcal{C})  ,\; \dot\Z\rangle_F).
\end{equation}
Then the unconstrained or Euclidean partial derivative matrix of the function $f$ at $\mathcal{C}$ with respect to codeword $\X$ is the matrix of components
$$
[D_{\X} f(\mathcal{C})]_{ab} = \frac{\partial f}{\partial x_{ab}} = \frac{\partial f}{\partial \Re(x_{ab})}+ i\,\frac{\partial f}{\partial \mathfrak{I}( x_{ab})},
$$
which we will be able to find by identifying the matrix $\dot\X$ in $\Re(\langle\;\dot\X  ,\; \dot\Z\rangle_F)$ when taking derivatives in the direction of $\dot\Z_{ab}$, for every matrix component $(a,b)$ of $\X$.


We summarize in Algorithm \ref{alg:RiemanMAC} a general Riemannian manifold optimization method for designing noncoherent MIMO MAC constellations for $K$ users.

\begin{algorithm}
\textbf{Input:} $\sum_{k=1}^K L_k$ uniformly distributed points in $\mathbb{C}^{T\times M}$
\\
\textbf{Output:} Optimized joint constellation $\mathcal{C}$
\begin{enumerate}
\item Choose cost function $f$ and manifold $\mathcal{M}\in\{\mathbb{G}(M,\mathbb{C}^T),\mathcal{OB}_{\mathbb{C}}(TM,\sum_k L_k),\mathbb{T}r(K,L,M,\mathbb{C}^T)\}$.

\item Compute unconstrained gradient $D_{\X_{k,i}}f(\mathcal{C})$ for every codeword $\X_{k,i}$ in $\mathcal{C}_k, (k=1,\dots,K, i=1,\dots, L_k)$.

\item Project down to the chosen manifold tangent space at every $\X$: $$\nabla_{\X}f(\mathcal{C})= \textsc{P}_{\X}(D_{\X}f(\mathcal{C})).$$ 

\item Compute the norm of the full gradient: $$||\nabla f(\mathcal{C})||=\sqrt{\sum_{k=1}^K\sum_{i=1}^{L_k}||\nabla_{\X_{k,i}}f(\mathcal{C})||^2_F }.$$

\item Move every codeword a step $h$ in the direction of steepest ascent (descent) retracting back onto the manifold: $$\X_{new} = \textsc{R}_{\X}\left(\pm h\frac{\nabla_{\X}f(\mathcal{C})}{||\nabla f(\mathcal{C})||}\right).$$

\item Evaluate $f(\mathcal{C}_{new})$ and repeat step $5$ with smaller $h$ until cost function improves its value with respect to $f(\mathcal{C})$.

\item Update constellation by substituting $\X\mapsto\X_{new}$ for every codeword.

\item Repeat $2-7$ until the number of iterations or improvement in $f$ reach a threshold.

\item Return constellation $\mathcal{C}_k =\{\X_{k,i}\}_{i=1}^{L_k}$, for every user $k=1,\dots, K$.
\end{enumerate}

\caption{Riemannian optimization for $K$-user MAC}
\label{alg:RiemanMAC}
\end{algorithm}


\section{Full-diversity noncoherent multiuser constellations for the MIMO MAC}\label{sec:ubpep}

\subsection{Noncoherent joint pairwise error probability}\label{subsec:asym_pep}

The study of coherent and noncoherent multiuser space-time communications was carried out extensively in \cite{Brehler01co} and \cite{Brehler01noco} using the results of \cite{Varanasi01}, where the asymptotic analysis of the error probability of quadratic receivers in Rayleigh fading channels was studied in detail. One of the important results of \cite{Brehler01noco} is that at least $T= (K+1)M$ temporal dimensions are necessary to achieve full-spatial diversity of $MN$ for every user. Recall that the spatial diversity indicates the slope of the SER vs. SNR curve when ${\rm SNR} \to \infty$. For $K=2$ users, this means that the coherence time must be at least $T=3M$ symbol periods.

Assuming full-diversity scenarios, Brehler and Varanasi derived in \cite{Brehler01noco} the asymptotic joint pairwise error probability of the ML detector in the noncoherent case. However, the PEP expression in \cite{Brehler01noco} has not been used as an optimization criterion so far as it was considered untractable for optimization. Further, the PEP expression was thought not to give clear insights for constellation design, as discussed in \cite{joint_ngo} and \cite{joint_guillaud}. In the present work we prove that the asymptotic PEP formula can not only be used to optimize joint constellations but actually provides the designs of choice for full-diversity scenarios. To the best of our knowledge, in this paper we provide for the first time exact formulas of its gradient on several manifolds with respect to every single-user codeword for any number of users. 

Let us introduce the following notation for the orthogonal projection matrix onto the orthogonal complement of the subspace spanned by the columns of $\M$:
$$
\P^\perp_{\M} = \I - \M(\M^{\He}\M)^{-1}\M^{\He}.
$$
Following \cite{Brehler01noco}, when comparing two joint hypothesis $\F_i$ vs. $\F_j$, the single-user codewords are to be reordered within the multiuser codeword so that the terms in error appear first, i.e., $\F_i=[\F_i^{(e)}\;\F^{(c)}]$ and $\F_j=[\F_j^{(e)}\;\F^{(c)}]$, where $\F^{(c)}$ are the codewords common to the two hypotheses or multiuser codewords, and $\F_i^{(e)},\;\F_j^{(e)}$ the codewords of the users in error between the two different hypothesis. 
With these conventions in place, the following proposition shows the expression derived in \cite{Brehler01noco} for the asymptotic PEP $\mathcal{P}(\F_i\rightarrow\F_j)$, i.e., the error probability in a binary hypothesis test between $\F_i$ and $\F_j$.

\begin{proposition}[Noncoherent Asymptotic Pairwise Error Probability\cite{Brehler01noco}]\label{prop:PEP}
Let us assume no correlation between the channel fading coefficients, equal SNR users\footnote{To consider users with different SNRs simply requires introducing a fixed diagonal matrix in the cost function. When there is spatial correlation at either the transmit or the receive arrays, it is necessary to include another matrix in the expression. These matrices are fixed and do not change the optimization process.}, and that $\F^{(e)H}_i\P^\perp_{\F_j}\F^{(e)}_i$ has full rank (i.e. $T \geq (e+K)M$, with $e$ the number of symbols in error). Then the total pairwise error probability of the optimal detector, for detecting $\F_j$ when receiving $\F_i$, approaches when the SNR grows arbitrarily closely to
\begin{equation}\label{eq:asynPEP}
    \mathcal{P}(\F_i\rightarrow\F_j) = \frac{\sigma^{2eNM}\sum_{n=0}^{eNM}\binom{2eNM-n}{eNM}(n!)^{-1}(\hat{c}_{ij})^n}{\det(\F^{(e)H}_i\P^\perp_{\F_j}\F^{(e)}_i )^N},
\end{equation}
where $\hat{c}_{ij} = N\log\frac{\F_i^{\He}\F_i}{\F_j^{\He}\F_j} \geq 0$, which can always be guaranteed by relabeling the hypothesis accordingly; $\sigma^2$ is the noise variance.
\end{proposition}


Notice that the denominator in $\mathcal{P}(\F_i\rightarrow\F_j)$ is the factor that encodes for the distance between joint codewords in error. 
This leads us to propose a multiuser union bound cost function for the design of noncoherent multiuser constellations:
\begin{equation}\label{eq:MACUB}
    f(\mathcal{C}) = \sum_{i\neq j} \sigma^{2eNM}\det(\F^{(e)H}_i\P^\perp_{\F_j}\F^{(e)}_i )^{-N},
\end{equation}
where the sum is over all the joint multiuser codewords in $\mathcal{C}$. 

\begin{remark}

The multiuser union bound (UB) criterion \eqref{eq:MACUB} is a natural generalization of the single-user ($K=1$) UB defined by
\begin{equation}\label{eq:singleUB}
{\mathrm{UB}}(\X_1,\ldots,\X_L) = \sum_{i < j}  \det \left( \I_M - \X_i^{\He}\X_j \X_j^{\He}\X_i \right)^{-N},
\end{equation}
which has been proposed in \cite{brehler02ub}, \cite{AlvarezEusipco22} to design single-user Grassmannian constellations. Optimized designs on the Grassmann manifold using this criterion have been obtained in our previous works \cite{AlvarezEusipco22,CuevasTCOM}. 

To see the connection between the single-user and the multiuser criteria, notice that when there is only one user present $\F_i = \X_i$ and $\P^\perp_{\F_j} = \I - \X_j\X_j^{\He}$, using $\X_i^{\He}\X_i=\X_j^{\He}\X_j=\I_M$, so up to a scaling constant
$$
f(\mathcal{C})_{K=1}= {\mathrm{UB}}(\X_1,\ldots,\X_L).
$$
\end{remark}

Notice that $e$ can take values from $1$ symbol in error to all the $K$ users in error, which makes the number of terms in the sum increasingly large: as the size of $|\mathcal{C}|=L_1\cdots L_K$ grows, the number of pairs of hypothesis $i,j$, i.e. number of terms in the sum (\ref{eq:MACUB}), grows as $\sim |\mathcal{C}|^2$. For example, for two users $K=2$, and $e=1$, there are $L_1(L_1-1)L_2+L_1L_2(L_2-1)$ terms in Eq. (\ref{eq:MACUB}), whereas for $e=2$ there are $L_1L_2(L_1-1)(L_2-1)$ terms, that is, the number of terms with two symbols in error grows with one order higher. This would make the multiuser optimization problem computationally unfeasible as the number of users and codewords grow. However, the contribution of the factor $\sigma^{2eNM}$ is $\sigma^{2NM}$ for the less numerous one-error terms and $\sigma^{4NM}$ for the more numerous two-error terms. Since $\sigma$ is inversely proportional to the SNR, the two-error terms are weighed two orders of magnitude less than the one-error terms. Because of this, and in order for the optimization to become feasible computationally, we propose to consider only the one-symbol-in-error terms, that is
\begin{equation}\label{eq:MACUBfinal}
F(\mathcal{C}) =  \sum_{\substack{\F_i\neq \F_j\in\mathcal{C} \\ e=1 }} \mathcal{F}_{ij}(\F_i^{(e)},\F_j)=  \sum_{\substack{\F_i\neq \F_j\in\mathcal{C} \\ e=1 }}\det(\F^{(e)\He}_i\P^\perp_{\F_j}\F^{(e)}_i )^{-N},
\end{equation}
so that the proposed design criterion for full-diversity scenarios finally becomes:
\begin{equation}\label{eq:MACUBfinal2}
\underset{\mathcal{C}_1,\, \dotsc,\, \mathcal{C}_K}{\argmin}\; F(\mathcal{C}).
\end{equation}
For example, for two users with $e=1$, let $\F_i=[\A\;\C]$ when the user $k\in\{1,2\}$ is in error, mistaking $\X_{k,i}=\A$ for $\X_{k,j}=\B$, whereas the other user detects the correct symbol $\C$, so that $\F_j=[\B\;\C]$, $\F_i^{(e)}=\A$ and $\F^c=\C$. In this case, each summand $\mathcal{F}_{ij}$ in the cost function (\ref{eq:MACUBfinal}) can be written explicitly in terms of the single-user codewords as $\mathcal{F}_{ij}(\A,\B,\C)$ so that:
\begin{align*}
 F(\mathcal{C})_{K=2} & = \sum_{\substack{\A\neq\C\in\mathcal{C}_1\text{ or }\mathcal{C}_2 \\ \B\in\mathcal{C}_2\text{ or }\mathcal{C}_1 }}\mathcal{F}_{ij}(\A, \B, \C) \\
 & = \sum_{\substack{\A\neq\C\in\mathcal{C}_1\text{ or }\mathcal{C}_2 \\ \B\in\mathcal{C}_2\text{ or }\mathcal{C}_1 }}\det(\A^{\He}(\I - [\B\;\C]([\B\;\C]^{\He}[\B\;\C])^{-1}[\B\;\C]^{\He})\A )^{-N} \\ & = \sum_{\substack{\A\neq\C\in\mathcal{C}_1\text{ or }\mathcal{C}_2 \\ \B\in\mathcal{C}_2\text{ or }\mathcal{C}_1 }}\det\left( \A^{\He}\A - \A^{\He}[\B\;\C]\cdot\begin{bmatrix} \B^{\He}\B & \B^{\He}\C \\ \C^{\He}\B & \C^{\He}\C\end{bmatrix}^{-1}\!\!\!\!\cdot[\B\;\C]^{\He}\A \right)^{-N}.
\end{align*}

In order to design codebooks based on minimizing the joint probability of error, we propose to perform a gradient descent algorithm over the packing $\mathcal{C}$ to minimize the cost function (\ref{eq:MACUBfinal}), for which we need the Riemannian gradient vector of $f$ in the Grassmannian product manifold or the oblique and trace manifolds described in Sec. \ref{sec:manifolds}. Notice that if one performs the optimization of the constellation within a submanifold other than the Grassmannian, like the oblique manifold, the single-user codewords need not be Stiefel matrices, so the terms $\A^{\He}\A,\;\B^{\He}\B,\;\C^{\He}\C$ of the last equation do not necessarily simplify to the identity, therefore the gradients of these type of functions must be computed without assuming these terms are the identity matrix at every step.

\subsection{Gradient computation}\label{subsec:pep_grad}

Let us write $\X = \F[i_1(\X):i_M(\X)]$ for the extraction of the $M$ columns in $\F$ running from column $i_1$ to column $i_M$ corresponding to the position of the codeword $\X$ inside the concatenated matrix $\F$. Then, for any other matrix $\A$ with the size of $\F$, $\A[i_1(\X):i_M(\X)]$ extracts the corresponding columns of $\A$ located where the block of $\X$ is within $\F$. With this and all the notational conventions defined in previous sections we arrive at the following fundamental result, that is actually valid for $F$ summed over any number of symbols in error.
\begin{theorem}\label{th:gradientBrehlerVaranasi}
Let $\M_j := \F_j^{\He}\F_j$ and $\G_{ij}:=\F_i^{(e)\He}\P^\perp_{\F_j}\F_i^{(e)}$. The unconstrained Euclidean gradient of $F$ with respect to codeword $\X$ is:
\begin{equation}
D_{\X}F(\mathcal{C}) =  \sum_{\substack{\F_i\neq \F_j\in\mathcal{C}}}D_{\X}\mathcal{F}_{ij}(\F_i^{(e)},\F_j),
\end{equation}
where for a codeword in error, $\X=\X^{(e)}$,  
the gradient matrix is given by the corresponding block of $M$ columns in the following expression
\begin{equation}
    D_{\X^{(e)}}\mathcal{F}_{ij}(\F_i^{(e)},\F_j) =  -2N\mathcal{F}_{ij}\,\left[\P^{\perp}_{\F_j}\F_i^{(e)}\G_{ij}^{-1}\right][i_1(\X^{(e)}):i_M(\X^{(e)})].
\end{equation}
And for $\X=\X^{(c)}$, a codeword not in error, 
we have:
\begin{equation}
D_{\X^{(c)}}\mathcal{F}_{ij} = 2N\mathcal{F}_{ij}\left[(\I_T -\F_j\M_j^{-1}\F_j^\He)\F_i^{(e)}\G_{ij}^{-1}\F_i^{(e)\He}\F_j\M_j^{-1}\right][j_1(\X^{(c)}):j_M(\X^{(c)})].
\end{equation}

\end{theorem}
\begin{proof}
Let us vary the components of $\F_j$ corresponding to a codeword $\X^{(c)}$ that is not in error, i.e. every component within $\F^{(c)}$, so that $d\F_i^{(e)}/dt\vert_{t=0}=0$. Let $\dot{\Z}_j$ be the matrix of the dimensions of $\F_j$ and with $0$ everywhere except $1$ at a fixed component, i.e. the unit variation of element $[\F_j]_{ab}$ for any row $a$ and column $b$ in the $\F^{(c)}$ part. We get the derivative of $\P^\perp_{\F_j}$ using the derivative of an inverse matrix which is
\begin{equation}\label{eq:derivativeinverse}
    \frac{d(\R(t)^{-1})}{dt}\bigg\vert_{t=0} = -\R(0)^{-1}\left(\frac{d\R(t)}{dt}\bigg\vert_{t=0}\right)\R(0)^{-1},
\end{equation}
and so
$$
\frac{d}{dt}\bigg\vert_{t=0}\P_{\F_j}^\perp(\F_j+t\dot{\Z}_j)= -\dot{\Z}_j\M_j^{-1}\F_j^\He-\F_j\M_j^{-1}\dot{\Z}_j^\He+\F_j\M^{-1}(\dot{\Z}_j^\He\F_j+\F_j^\He\dot{\Z}_j)\M^{-1}\F_j^\He.
$$
Thus, for the derivative of $\mathcal{F}_{ij}$ when varying only this element in the codewords of $\mathcal{C}$, and using the derivative of the determinant formula by Jacobi, one obtains
\begin{align*}
& \frac{d}{dt}\bigg\vert_{t=0}\mathcal{F}_{ij}(\F_i^{(e)},\F_j+t\dot{\Z}_j)= -N(\det\G_{ij})^{-N-1}\det\G_{ij}\tr\left[\G_{ij}^{-1}\frac{d}{dt}\bigg\vert_{t=0}\G_{ij}(\F_j+t\dot{\Z}_j)\right]
\\ 
\\
& = -N\mathcal{F}_{ij}\tr\left[ \G_{ij}^{-1}\F_i^{(e)\He}\frac{d\P_{\F_j}^\perp}{dt}\bigg\vert_{t=0}\F_i^{(e)} \right] \\ \\
& = -N\mathcal{F}_{ij}\tr\left[ \G_{ij}^{-1}\F_i^{(e)\He}(-\dot{\Z}_j\M_j^{-1}\F_j^\He-\F_j\M_j^{-1}\dot{\Z}_j^\He)\F_i^{(e)} \right] + \\
& \hspace{4.5cm} -N\mathcal{F}_{ij}\tr\left[ \G_{ij}^{-1}\F_i^{(e)\He}\F_j\M^{-1}(\dot{\Z}_j^\He\F_j+\F_j^\He\dot{\Z}_j)\M^{-1}\F_j^\He\F_i^{(e)} \right].
\end{align*}
As always, using the cyclic property of the trace, that $\P^\perp_{\F_j}$ and $\G_{ij}$ are Hermitian, and the definition of Frobenius inner product, this reduces to
\begin{align*}
\frac{d}{dt}\bigg\vert_{t=0}\mathcal{F}_{ij}(\F_i^{(e)},\F_j+t\dot{\Z}_j) = &
\; 2N\mathcal{F}_{ij}\;\Re\langle\; \F_i^{(e)}\G^{-1}_{ij}\F_i^{(e)\He}\F_j\M^{-1}_j,\;\dot{\Z}_j\; \rangle_F\; + \\ \quad\quad\quad & - 2N\mathcal{F}_{ij}\;\Re\langle\; \F_j\M_j^{-1}\F_j^\He\F_i^{(e)}\G_{ij}^{-1}\F_i^{(e)\He}\F_j\M_j^{-1},\; \dot{\Z}_j\; \rangle_F,
\end{align*}
which implies that the partial derivative with respect to that matrix element is
$$
\frac{\partial\mathcal{F}_{ij}}{\partial x^{(c)}_{ab}} = 2N\mathcal{F}_{ij}\left[ \F_i^{(e)}\G^{-1}_{ij}\F_i^{(e)\He}\F_j\M_j^{-1}  -  \F_j\M_j^{-1}\F_j^\He\F_i^{(e)}\G_{ij}^{-1}\F_i^{(e)\He}\F_j\M_j^{-1} \right]_{ab}.
$$
Therefore, collecting terms, the unconstrained gradient of $\mathcal{F}_{ij}$ with respect to a codeword $\X^{(c)}$ in $\F^{(c)}$, corresponding to the columns $[j_1(\X^{(c)}):j_M(\X^{(c)})]$ of $\F_j$, is:
$$
\left[\frac{\partial\mathcal{F}_{ij}}{\partial \X^{(c)}}\right] =\dot{\X}^{(c)} = 2N\mathcal{F}_{ij}\left[(\I_T - \F_j\M_j^{-1}\F_j^\He)\F_i^{(e)}\G_{ij}^{-1}\F_i^{(e)\He}\F_j\M_j^{-1}\right][j_1(\X^{(c)}):j_M(\X^{(c)})].
$$
Similarly, one can vary the components of $\F_i^{(e)}$ corresponding to the codeword $\X^{(e)}$ in error, so that $d\F_j/dt\vert_{t=0}=0$, and therefore:
\begin{align*}
& \frac{d}{dt}\bigg\vert_{t=0}\mathcal{F}_{ij}(\F_i^{(e)}|\X^{(e)}+t\dot{\Z},\F_j) = -N\mathcal{F}_{ij}\tr\left[ \G_{ij}^{-1}\left( \dot{\Z}^\He_i\P^\perp_{\F_j}\F_i^{(e)}+ \F_i^{(e)\He}\P^\perp_{\F_j}\dot{\Z} \right) \right] \\ \\
& = -N\mathcal{F}_{ij}\tr\left[ \dot{\Z}^\He\P^\perp_{\F_j}\F_i^{(e)}\G_{ij}^{-1} \right] -N\mathcal{F}_{ij}\tr\left[ \G_{ij}^{-1}\F_i^{(e)\He}\P^\perp_{\F_j}\dot{\Z} \right] \\ \\
& = -2N\mathcal{F}_{ij}\;\Re\langle\; \P^\perp_{\F_j}\F_i^{(e)}\G_{ij}^{-1},\;\dot{\Z} \;\rangle_F,
\end{align*}
which yields
$$
\frac{\partial\mathcal{F}_{ij}}{\partial x^{(e)}_{ab}} = -2N\mathcal{F}_{ij}\left[ \P^\perp_{\F_j}\F_i^{(e)}\G_{ij}^{-1}\right]_{ab},
$$
and finally
$$
\left[\frac{\partial\mathcal{F}_{ij}}{\partial \X^{(e)}}\right] =\dot{\X}^{(e)} = -2N\mathcal{F}_{ij}\left[ \P^\perp_{\F_j}\F_i^{(e)}\G_{ij}^{-1}\right][i_1(\X^{(e)}):i_M(\X^{(e)})].
$$
\end{proof}


\section{Non-full diversity multiuser constellations for the MIMO MAC}\label{sec:proxies}

\subsection{Cost functions}\label{subsec:proxies}
Obviously, whenever the coherence time is sufficiently long such that $T \geq (K+1)M$, it will be preferable to design and employ full-diversity multiuser constellations. However, the constraint $T \geq (K+1)M$ may be difficult to meet in fast-fading channels with high-mobility users, especially when the number of users or transmitting antennas grows.

In the non-full diversity case, the full-rank condition required in Proposition \ref{prop:PEP} is not satisfied and the asymptotic PEP formula is no longer correct in this scenario. However, one can use instead proxy functions that provide bounds for the PEP valid without full diversity. Several cost functions have been proposed to design constellations for noncoherent communications in the multiple access channel, cf. \cite{Ngo18} or \cite{ngoyang}. In particular, the authors of \cite{ngoyang} introduce several cost functions to design multiuser constellations on Grassmannians, such as the functions named $\beta, \delta$ and $J_{1/2}$ to be presented below. These depend on pairs of joint codewords $\F_i,\F_j\in\mathcal{C}$, and are all related to the leading exponent of the joint pairwise error probability  $\mathcal{P}(\F_i\rightarrow\F_j)$, providing bounds that serve as proxies for this PEP. Although the criteria in \cite{ngoyang} were proposed for both full-diversity and non-full diversity scenarios, our experience indicates that for full-diversity scenarios with $T \geq (K+1)M$ the criterion based on the asymptotic expression of the PEP, described in Subsection \ref{sec:ubpep}, provides much better results. The criteria described in this section are thus specifically useful to design non-full diversity multiuser constellations.

The main geometrical motivation to study these cost functions however stems from the fact that they are all related to a geometrical interpretation of $\delta$ as a Riemannian distance between Hermitian positive definite matrices defined from the joint codewords $(\I_T+\F_i\F_i^{\He})$ and $(\I_T+\F_j\F_j^{\He})$. Thus, in \cite{ngoyang}, design criteria are proposed that maximize these cost functions, due to their relation to the worst PEP and the intuition of separating the closest pair of joint codewords in the manifold of $T\times T$ Hermitian positive definite matrices. However, the authors in \cite{ngoyang} only optimize the $\max-J_{1/2,\min}$ for USTM single-user codewords, i.e. they optimized $\max-J_{1/2,\min}$ in the Grassmann manifold. In the following we work out the theoretical basis needed to optimize  a union-bound-based generalization of the cost functions $\beta$ and $\delta$, and to do so on the different manifolds presented in Subsection \ref{sec:manifolds}.

This leads us to propose the following cost functions for the design of noncoherent multiuser constellations for the MIMO MAC in non-full diversity scenarios.

\begin{definition}
The SER union bound proxy function \emph{beta} for a joint constellation $\mathcal{C}$ is defined as
\begin{equation}\label{eq:betaub}
\beta_{UB}(\mathcal{C}) :=\log\left[ \sum_{\mathbf{F}_i\neq\mathbf{F}_j\in\mathcal{C}}\exp\left(-N\sum_{l=1}^T|\log\,\lambda_l(\F_i,\F_j)|\right)\right],
\end{equation}
where $\lambda_l(\X,\Y)$ are the eigenvalues of $\GG(\F_i,\F_j) := (\I_T +\F_i\F_i^{\He})(\I_T +\F_j\F_j^{\He})^{-1}$, for $\F_i,\; \F_j\in\mathcal{C}$ multiuser codewords.
\end{definition}
Notice that $\lambda_l\geq 0$, and except for a subset of measure zero in the space of matrices $\F_i,\F_j$ the eigenvalues will be positive so that the logarithm is well defined.
Equivalently, we can use the pairwise hypothesis function
$$
\beta(\F_i,\F_j) := \sum_{l=1}^T|\log\,\lambda_l(\F_i,\F_j)|,\quad
\text{ so that }\quad \beta_{UB}(\mathcal{C}) =\log\left[ \sum_{\mathbf{F}_i\neq\mathbf{F}_j\in\mathcal{C}}\exp[-N \beta(\F_i,\F_j)]\right].
$$
\begin{definition}
The SER union bound proxy function \emph{delta} for a joint constellation $\mathcal{C}$ is defined as
\begin{equation}\label{eq:deltaub}
\delta_{UB}(\mathcal{C}) :=\log\left[ \sum_{\mathbf{F}_i\neq\mathbf{F}_j\in\mathcal{C}}\exp\left(-N\sqrt{\sum_{l=1}^T\log^2\,\lambda_l(\F_i,\F_j)}\right)\right],
\end{equation}
where $\lambda_l(\F_i,\F_j)$ are the eigenvalues of $\GG(\F_i,\F_j)$ as above, for $\F_i,\, \F_j\in\mathcal{C}$.
\end{definition}
Similar notation to the former cost function leads us to write
$$
\delta(\F_i,\F_j) := \sqrt{\sum_{l=1}^T\log^2\,\lambda_l(\F_i,\F_j)},\quad
\text{ so that }\quad \delta_{UB}(\mathcal{C}) =\log\left[ \sum_{\mathbf{F}_i\neq\mathbf{F}_j\in\mathcal{C}}\exp[-N \delta(\F_i,\F_j)]\right].
$$

The motivation for these comes from the following results from \cite[Prop. 3 and 4]{ngoyang} which relate the pairwise function we just defined to the exponent of the joint probability of error.

\begin{proposition}
\label{prop:beta}
The joint PEP exponent is upper- and lower-bounded as
\begin{equation}\label{eq:betaIneqPEP}
    \beta(\F_i,\F_j) + T \geq -\frac{1}{N}\log\mathcal{P}(\F_i\rightarrow\F_j)\geq \frac{1}{2}\beta(\F_i,\F_j) - T\log 2.
\end{equation}
\end{proposition}
\begin{proposition}\label{prop:delta}
The natural Riemannian distance between $\I_T +\F_i\F_i^{\He}$ and $\I_T +\F_j\F_j^{\He}$, in the manifold of Hermitian positive definite matrices, is $\delta(\F_i,\F_j)$, and $\beta(\F_i,\F_j)$ is bounded by it as
\begin{equation*}
    \sqrt{T}\delta(\F_i,\F_j)\geq \beta(\F_i,\F_j)\geq\delta(\F_i,\F_j).
\end{equation*}
\end{proposition}

By multiplying by $-N$ and taking exponentials, Eq. (\ref{eq:betaIneqPEP}) becomes
$$
\exp[-N\beta(\F_i,\F_j)]\exp[ - NT] \leq \mathcal{P}(\F_i\rightarrow\F_j)\leq \exp\left[\frac{-N}{2}\beta(\F_i,\F_j)\right]\exp[ NT\log 2],
$$
Since $\delta$ is in turn a bound on $\beta$, these relations lead us to expect a leading order behavior such as
$$
\mathcal{P}(\F_i\rightarrow\F_j) \sim \exp[-N\beta(\F_i,\F_j)]\quad\text{ and } \quad\mathcal{P}(\F_i\rightarrow\F_j) \sim\exp[-N\delta(\F_i,\F_j)],
$$
and by summing over all pairs of joint codewords and taking logarithms, the inequality yields
$$
e^{- NT}\sum_{\F_i\neq \F_j}\exp[-N\beta(\F_i,\F_j)] \leq \sum_{\F_i\neq \F_j}\mathcal{P}(\F_i\rightarrow\F_j)\leq e^{NT\log 2}\sum_{\F_i\neq \F_j}\exp\left[\frac{-N}{2}\beta(\F_i,\F_j)\right].
$$
This provides bounds on the union bound PEP and on its leading exponent by taking logarithms, which justify our definition of $\beta_{UB}$ in Eq. (\ref{eq:betaub}) as a figure of merit to optimize (the factor $1/2$ is irrelevant in the normalized gradients). Similarly for $\delta_{UB}$.

Therefore the criteria that we propose are
\begin{equation}
    \underset{\mathcal{C}_1, \dotsc, \mathcal{C}_K}{\argmin}\;\delta_{UB}(\mathcal{C}),
\end{equation}
and
\begin{equation}
    \underset{\mathcal{C}_1, \dotsc, \mathcal{C}_K}{\argmin}\;\beta_{UB}(\mathcal{C}),
\end{equation}
where $\mathcal{C}$ is built up out of the concatenation of $K$ codewords $\X_{k,i}\in\mathcal{C}_k$, one from each users's constellation, which is a set of points in the corresponding manifold.

For completeness, since it is used in the comparisons of our simulation experiments in Sec. \ref{sec:results}, we also introduce another cost function proposed and optimized in \cite{ngoyang}:
\begin{equation}
    J_{1/2}(\F_i,\G_j) = \frac{1}{2}\log\det( 2\I_T + (\I_T +\F_j\F_j^{\He})^{-1}(\I_T +\F_i\F_i^{\He}) + (\I_T +\F_i\F_i^{\He})^{-1}(\I_T +\F_j\F_j^{\He}) ) -T\log 2.
\end{equation}

\subsection{Gradient computation}\label{subsec:proxy_grad}

In order to perform a gradient descent on the cost functions defined above, we need to compute their derivatives along directions tangent to the manifolds of interest. First of all, we must ensure that the $\lambda_l(\F_i,\F_j)$ functions are smooth almost everywhere, which can be proved on any chart by using the locally Lipschitz property and Rademacher's theorem \cite[Th. 3.1.6]{federer}. Even though this guarantees that the gradients are well-defined except for a subset of measure zero, at every iteration of the algorithm, the explicit computation of the derivatives for our manifolds of interest is essentially intractable using manifold charts. Instead of differentiating the restricted functions on the manifold, we use a simpler method based on differentiation on the ambient space and then projecting down the unrestricted gradient to the manifold, which is justified by the excellent optimization results it provides (see Sec. \ref{subsec:resultsProxies}). This method relies on the following fundamental result (notice that multiple eigenvalues happen only for a subset of measure zero in the space of matrices).

\begin{lemma}\label{lem:eigendiff}
 Let $\GG(t)\in\mathbb{C}^{T\times T}$ be a matrix function with eigenvalue system $\GG(0)\mathbf{v}_0 = \lambda_0\mathbf{v}_0$, such that $\lambda_0$ is a simple eigenvalue with associated eigenvector $\mathbf{v}_0$. Then there are functions $\lambda(t),\mathbf{v}(t)$ defined for all $\GG$ in a neighbourhood of $\GG(0)$ such that $\lambda(0)=\lambda_0,\; \mathbf{v}(0)= \mathbf{v}_0$ and $\GG(t)\mathbf{v}(t) = \lambda(t)\mathbf{v}(t)$, with $\mathbf{v}_0^{\He}\mathbf{v}(t)=1$. Moreover, these functions are infinitely differentiable with derivatives:
\begin{equation}
    \frac{d\lambda}{dt}\biggr\vert_{t=0}= \frac{1}{\u^{\He}_0\mathbf{v}_0}\u^{\He}_0\cdot \frac{d\GG}{dt}\biggr\vert_{t=0}\cdot\mathbf{v}_0  
\end{equation}
and
\begin{equation}
     \frac{d\mathbf{v}}{dt}\biggr\vert_{t=0} = (\lambda\I_T -\GG)^+\left(\I_T -\frac{\mathbf{v}_0\u^{\He}_0}{\u^{\He}_0\mathbf{v}_0} \right)\frac{d\GG}{dt}\biggr\vert_{t=0}\cdot\mathbf{v}_0 
\end{equation}
where $\u_0$ is the eigenvector associated to the eigenvalue $\lambda^*_0$ of $\GG^{\He}(0)$, i.e. $\GG^{\He}(0)\u_0=\lambda^*_0\u_0$, and $(\cdot)^+$ is the Moore-Penrose pseudo-inverse.
\end{lemma}
\begin{proof}
See, e.g., \cite[theorem 2]{magnus}.
\end{proof}

With all these tools and the notation conventions of the previous section, we can derive the explicit gradients of any of the proposed cost functions for the $K$-user MIMO MAC.

\begin{theorem}
\label{th:graddelta}
When  $\GG_{ij}:=\GG(\F_i,\F_j)$ satisfy the conditions of lemma \ref{lem:eigendiff} for every $\F_i,\F_j\in\mathcal{C}$, the unconstrained partial derivative of the function $\delta_{UB}$, with respect to the single-user codeword $\X$, is  given by
\begin{equation}
       D_{\X}\,\delta_{UB}(\mathcal{C}) = \sum_{\mathbf{F}_i\neq\mathbf{F}_j\in\mathcal{C}}\frac{-N\exp[-N\delta(\F_i,\F_j)]}{\exp[\delta_{UB}(\mathcal{C})]\delta(\F_i,\F_j)}\sum_{l=1}^T\frac{\log(\lambda_l)}{\lambda_l\mathbf{u}^{\He}_l\mathbf{v}_l}\mathbf{u}^{\He}_l\cdot \frac{\partial\mathbf{\Gamma}_{ij}}{\partial\X}\cdot\mathbf{v}_l,
\end{equation}
where the $\lambda_l$, $\u_l$ and $\mathbf{v}_l$ are for $\GG_{ij}$ as defined in Lemma \ref{lem:eigendiff}, and $\partial\mathbf{\Gamma}_{ij}/\partial\X$ is a matrix of matrices that depend on whether $\F_i$ or $\F_j$ include, one or both, the single-user symbol $\X$ as a block of its columns. In detail, let $\EE_{mn}$ be the canonical real matrix basis, i.e. $[\EE_{mn}]_{ij}=[\delta_{mi}\delta_{nj}]_{ij}$, for $i,j,m,n$ indices from $1$ to $T$, then:
\begin{itemize}
\item If $\X$ is the symbol in $\F_i$ of an user in error (i.e. a differing block with respect to $\F_j$):
\begin{align}
    \left[\frac{\partial\mathbf{\Gamma}_{ij}}{\partial\X}\right]_{mn} & = (\EE_{mn}\X^{\He} + \X\EE^{\He}_{mn})(\I_T + \F_j\F_j^{\He})^{-1}.
\end{align}
\item If $\X$ is the symbol in $\F_j$ of an user in error (i.e. a differing block with respect to $\F_i$):
\begin{align}
    \left[\frac{\partial\mathbf{\Gamma}_{ij}}{\partial\X}\right]_{mn} & = -\GG_{ij}(\EE_{mn}\X^{\He} + \X\EE^{\He}_{mn})(\I_T + \F_j\F_j^{\He})^{-1}.
\end{align}

\item If $\X$ is not a symbol in error (i.e. a common block between $\F_i$ and $\F_j$):
\begin{align}
    \left[\frac{\partial\mathbf{\Gamma}}{\partial\X}\right]_{mn} = (\I_T - \GG_{ij})(\EE_{mn}\X^{\He} + \X\EE^{\He}_{mn})(\I_T + \F_j\F_j^{\He})^{-1}.
\end{align}
\end{itemize}
\end{theorem}

\begin{proof}
Only the terms in the sum that correspond to joint codewords $\F_i$ or $\F_j$ containing $\X$ are nonzero in the derivative, hence we have
$$
\frac{d}{dt}\bigg\vert_{t=0}\delta_{UB}(\mathcal{C}(t)\vert\X+t\dot{\Z}) = \sum_{\substack{\F_i\neq \F_j\in\mathcal{C} \\ \X\in\F_i\text{ or }\F_j }}\frac{-N\exp[-N \delta(\F_i,\F_j)]}{\exp[\delta_{UB}(\mathcal{C})]}\cdot\,\frac{d}{dt}\bigg\vert_{t=0}\delta(\F_i,\F_j\vert\X+t\dot\Z),
$$
where
$$
\exp[\delta_{UB}(\mathcal{C})] = \sum_{\mathbf{F}_i\neq\mathbf{F}_j\in\mathcal{C}}\exp[-N \delta(\F_i,\F_j)].
$$
Notice that one must pay attention to whether the symbol $\X$ is included in the multiuser codeword $\F_i$, $\F_j$ or both, so each term in the previous sum should actually be written $\delta(\F_i\vert\X+t\dot\Z,\F_j)$, $\delta(\F_i,\F_j\vert\X+t\dot\Z)$ or a variation on both depending on the case for each summand.
Now, without loss of generality
$$
\frac{d}{dt}\bigg\vert_{t=0}\delta(\F_i,\F_j\vert\X+t\dot\Z) = \frac{1}{2}\left[\sum_{l=1}^T(\log\lambda_l)^2\right]^{-\frac{1}{2}}\sum_{l=1}^T 2\frac{\log\lambda_l}{\lambda_l}\;\frac{d}{dt}\bigg\vert_{t=0}\lambda_l(\F_i,\F_j\vert\X+t\dot\Z),
$$
and, using the notation and discussion from Lemma \ref{lem:eigendiff}, we know that in the generic case we can differentiate the eigenvalues with respect to every component of the matrix $\X$, in the direction $\dot\Z_{mn}$, to obtain
$$
\frac{d}{dt}\bigg\vert_{t=0}\lambda_l(\F_i,\F_j\vert\X+t\dot\Z_{mn}) = \frac{1}{\mathbf{u}^{\He}_l\mathbf{v}_l}\mathbf{u}^{\He}_l\cdot \left[\frac{\partial\mathbf{\Gamma}(\F_i,\F_j)}{\partial \X}\right]_{mn}\cdot\mathbf{v}_l.
$$
Notice that for any joint codeword, e.g. $\F_i=[\X_{1,i_1}\dots\X\dots\X_{K,i_K}]$, where our chosen $\X$ corresponds to some symbol $i_r$ of some user $r$, i.e. $\X=\X_{r,i_r}\in\mathcal{C}_r$, then the factors of $\GG_{ij}$ expand by blocks as
$$
\F_i\F_i^{\He} = \X\X^{\He} + \sum_{k\neq r}^{K} \X_{k,i_k}\X_{k,i_k}^{\He},
$$
where the block $\X$ is the symbol of the user of interest, and similarly for $\F_j$. Thus we must take derivatives of $\GG_{ij}$ with respect to an $\X$ that either belongs to only one of the sums of the above expansions of $\F_i\F_i^{\He}$ or $\F_i\F_i^{\He}$, or to both, which then requires the derivative of a product.

Let us assume first that $\X$ is a symbol in $\F_i$ of an user in error, i.e. it does not appear in $\F_j$, then the partial derivative of $\GG_{ij}$ with respect to the component $(m,n)$ of the symbol $\X$ is 
\begin{align*}
\left[\frac{\partial\mathbf{\Gamma}(\F_i,\F_j)}{\partial\X}\right]_{mn} & = \frac{d}{dt}\bigg\vert_{t=0}\mathbf{\Gamma}(\F_i\vert\X+t\dot\Z_{mn},\F_j) \\
\\
& = \frac{d}{dt}\bigg\vert_{t=0}\left(\I_T +(\X+t\dot\Z_{mn})(\X^{\He}+t\dot\Z_{mn}^{\He}) + \sum_{k\neq r}^{K} \X_{k,i_k}\X_{k,i_k}^{\He}\right)(\I_T +\F_j\F_j^{\He})^{-1} \\
\\
& = (\dot\Z_{mn}\X^{\He} + \X\dot\Z_{mn}^{\He})(\I_T +\F_j\F_j^{\He})^{-1},
\end{align*}
which yields the first of the formulas of the theorem. Similarly, when $\X$ is a symbol in $\F_j$ of an user in error, i.e. not appearing in $\F_i$, one obtains
\begin{align*}
& \left[\frac{\partial\mathbf{\Gamma}(\F_i,\F_j)}{\partial\X}\right]_{mn} = \frac{d}{dt}\bigg\vert_{t=0}\mathbf{\Gamma}(\F_i,\F_j\vert\X+t\dot\Z_{mn}) \\
\\
& = \frac{d}{dt}\bigg\vert_{t=0}(\I_T +\F_i\F_i^{\He})\left(\I_T +(\X+t\dot\Z_{mn})(\X^{\He}+t\dot\Z_{mn}^{\He}) + \sum_{k\neq r}^{K} \X_{k,j_k}\X_{k,j_k}^{\He}\right)^{-1} 
\end{align*}
\begin{align*}
& = (\I_T +\F_i\F_i^{\He})(-1)(\I_T +\F_j\F_j^{\He})^{-1}\frac{d}{dt}\bigg\vert_{t=0}\left(\I_T +(\X+t\dot\Z_{mn})(\X^{\He}+t\dot\Z_{mn}^{\He}) + \sum_{k\neq r}^{K-1} \X_{k,j_k}\X_{k,j_k}^{\He}\right) (\I_T +\F_j\F_j^{\He})^{-1} \\
\\
& = -\GG_{ij} (\dot\Z_{mn}\X^{\He} + \X\dot\Z_{mn}^{\He})(\I_T +\F_j\F_j^{\He})^{-1},
\end{align*}
which provides the second formula of the theorem.
In the final case, $\X$ is a common block between $\F_i$ and $\F_j$ corresponding to a symbol of an user not in error, and therefore it appears in both terms of $\GG_{ij}$ so that
\begin{align*}
& \left[\frac{\partial\mathbf{\Gamma}(\F_i,\F_j)}{\partial\X}\right]_{mn} = \frac{d}{dt}\bigg\vert_{t=0}\mathbf{\Gamma}(\F_i\vert\X+t\dot\Z_{mn},\F_j\vert\X+t\dot\Z_{mn}) \\
\\
& = \frac{d}{dt}\bigg\vert_{t=0}\!\!\!\left(\I_T +(\X+t\dot\Z_{mn})(\X^{\He}+t\dot\Z_{mn}^{\He}) +\!\! \sum_{k\neq r}^{K} \X_{k,i_k}\X_{k,i_k}^{\He}\right)\!\!\!\left(\I_T +(\X+t\dot\Z_{mn})(\X^{\He}+t\dot\Z_{mn}^{\He}) +\!\! \sum_{k\neq r}^{K} \X_{k,j_k}\X_{k,j_k}^{\He}\right)^{-1}
\end{align*}
which by the product rule results in the sum between the two formulas of the other cases obtained above, yielding the last equation of the theorem.
\end{proof}

\begin{theorem}
Using the same conventions and assumptions as in Theorem \ref{th:graddelta}, the gradient of the union bound proxy function $\beta_{UB}$ is given by:
\begin{equation}
       D_{\X}\,\beta_{UB}(\mathcal{C}) = \sum_{\mathbf{X}\neq\mathbf{Y}\in\mathcal{C}}\frac{-N\exp[-N\beta(\F,\G)]}{\exp[\beta_{UB}(\mathcal{C})]}\sum_{i=1}^T\frac{\text{sign}(\log\lambda_i)}{\lambda_i\mathbf{u}^{\He}_i\mathbf{v}_i}\mathbf{u}^{\He}_i\cdot \frac{\partial\mathbf{\Gamma}(\F,\G)}{\partial\X}\cdot\mathbf{v}_i,
\end{equation}
for $\X = \A,\B,\C,\D$ from $\F =[\A\;\B],\;\G =[\C\;\D]\in\mathcal{C}_1\times\mathcal{C}_2$.

\end{theorem}
\begin{proof}
The proof is exactly analogous to the previous theorem, simply taking into account a different chain rule that also results in terms of the gradient of $\mathbf\Gamma$. It is worth noticing that we can make the expression with the absolute value of the logarithm to become differentiable by considering $\beta$ given in terms of $|\log\lambda_i|^{1+\epsilon}$ instead, for small values of $\epsilon$.

\end{proof}

\section{Results}\label{sec:results}

\subsection{Noncoherent multiuser constellation designs in non-full diversity scenarios}\label{subsec:resultsProxies}

We first assess the performance of the multiuser designs for the MIMO MAC proposed in Sec. \ref{sec:proxies}, which are obtained by optimizing the cost functions $\beta_{UB}$ and $\delta_{UB}$, see Eq. (\ref{eq:betaub}) and Eq. (\ref{eq:deltaub}), respectively. We also take into account the different power normalization constraints mentioned in Sec. \ref{sec:manifolds}, so that the optimization is carried out over the corresponding submanifolds, a feature that yields outstanding differences in the case of non-full diversity. For ease of exposition, the simulations are restricted to a 2-user MIMO MAC. In every case the SER performance refers to the average performance of the two users after applying at the BS the ML multiuser detector.

In Fig. \ref{fig:Manifolds_T5M2N3L16_beta_delta} the performance of constellations designed over different manifolds is shown for $T=5$ symbol periods, $M=2$ transmit antennas, $N=3$ receive antennas, and $L=16$ codewords for each user. One can immediately appreciate that both cost functions perform better when the power constraint is relaxed: the USTM codewords from the Grassmann manifold are the worst performing designs, whereas those that constrain the average power are the best performing ones, with the per-codeword power constrained designs yielding an intermediate performance. The oblique and trace manifold constraints result in constellations with performance improvement ranging from half an order to an order of magnitude better SER. Moreover, the $\delta_{UB}$ criterion clearly outperforms the $\beta$ criterion significantly, except for the Grassmann manifold. Overall,  Fig. \ref{fig:Manifolds_T5M2N3L16_beta_delta} shows that in the non-full diversity case the manifold on which the optimization is carried out plays a very significant role. In other words, using codewords with different powers may produce codebooks with much better performance in multiuser scenarios.

\begin{figure}[H]
    \centering
\includegraphics[width=.6\textwidth]{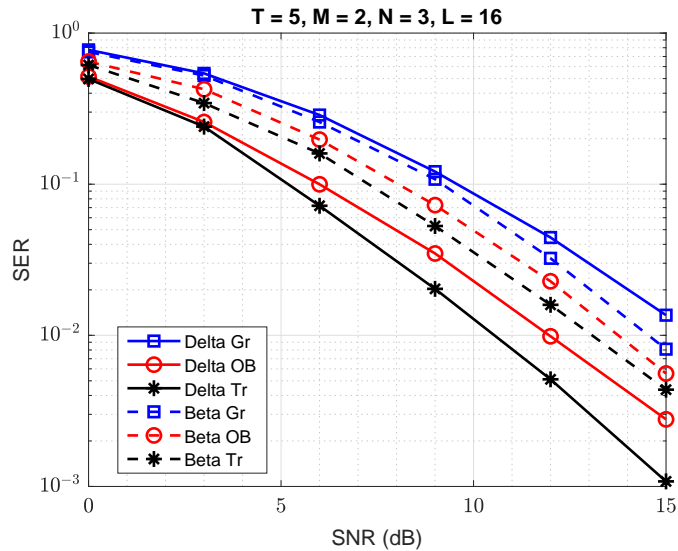}
     \caption{Multiuser codebook performance of $\beta$ and $\delta$ optimization designs in different Riemannian manifolds. The scenario represents a 2-user MIMO MAC with $T=5$, $M=2$, $N=3$ and $L=16$. }
	\label{fig:Manifolds_T5M2N3L16_beta_delta}
\end{figure}


Using the same settings as in the previous case, Fig. \ref{fig:Compare_T5M2N3L16_beta_delta} shows how the trace manifold design (the best performing design from the manifold comparison), outperforms the single-user designs and the multiuser $J_{1/2}$ criterion studied in \cite{ngoyang}. In particular, we have used single-user packings optimized using the minimum chordal distance and the coherence criterion, as explained in \cite{AlvarezEusipco22} and \cite{Cuevas21WSA}. The figure depicts another coherence constellation created by optimizing a concatenated single-user constellation of double size, and then splitting it up into two, one half for each user. This design is labeled as {\it Coherence Split} in the figure. One would expect that this would help to optimize the cross terms in the union bound of the single-user packings, but the result shows that this type of splitting performs even slightly worse than the plain single-user coherence optimization. Finally, for this setup we are not able to distinguish the $J_{1/2}$ performance from the performance of single-user designs, since the former only outperforms slightly the latter at low SNR.

\begin{figure}[H]
    \centering
\includegraphics[width=.6\textwidth]{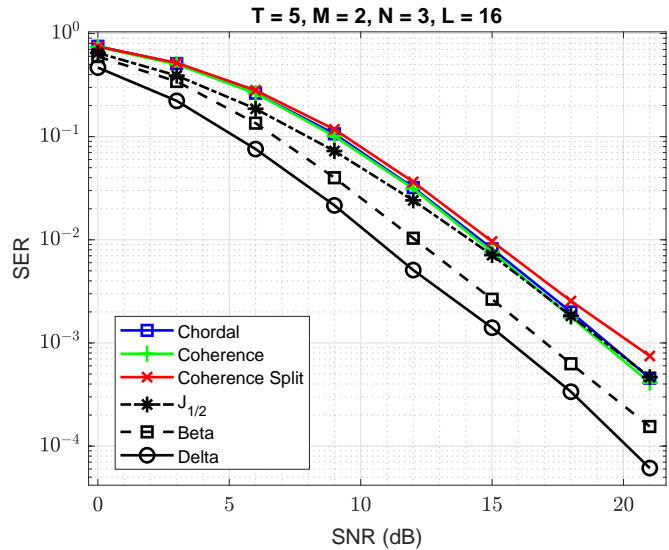}
     \caption{Multiuser codebook performance for jointly optimized designs ($J_{1/2},\beta,\delta$), for a 2-user MIMO MAC with $T=5$, $M=2$, $N=3$ and $L=16$, vs. single user designs based on the chordal and coherence criteria. }
	\label{fig:Compare_T5M2N3L16_beta_delta}
\end{figure}


In Fig. \ref{fig:Manifolds_T4M2N10L16_beta_delta} and Fig. \ref{fig:Compare_T4M2N10L16_beta_delta} the same type of analysis was carried out but for a different scenario, with $T=4$ symbol periods, $M=2$ transmit antennas, and a high number of receive antennas $N=10$. In Fig. \ref{fig:Manifolds_T4M2N10L16_beta_delta} the improvement in performance by using the oblique and trace manifolds is outstanding. Moreover, the optimization in the Grassmannian produces designs which are indistinguishable between both cost functions and even show a noise floor at high SNR. With this high number of receive antennas, the single-user designs also show a noise floor of similar magnitude as for the multiuser USTM designs. Even the multiuser cost function $J_{1/2}$ does not perform well, although the floor is lower than in the other cases. However, the performance of the $\delta_{UB}$ and $\beta_{UB}$ constellations is orders of magnitude better, without showing any noise floor, and it seems to attain a big portion of the full-diversity. Since these multiuser designs are a union bound of $\beta$ and $\delta$ functions, weighting the exponential on the number of receive antennas $N$, it seems natural that this parameter plays a role in the designs, as these figures confirm in comparison with the previous two.

\begin{figure}[H]
    \centering
\includegraphics[width=.6\textwidth]{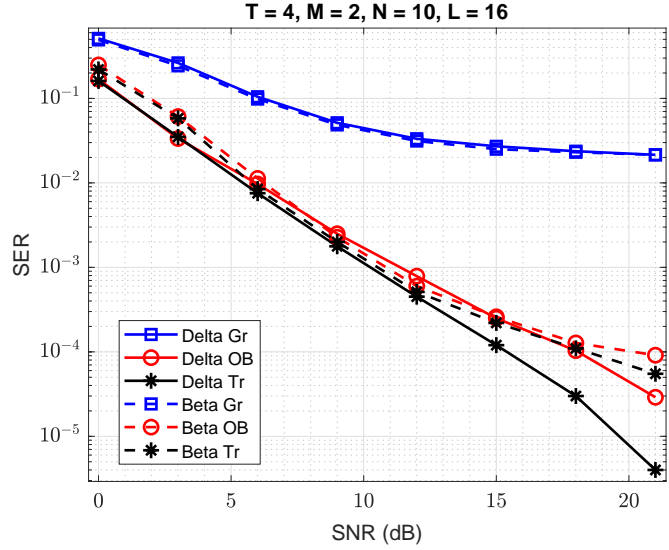}
     \caption{Performance of $\beta$ and $\delta$ optimization designs in different Riemannian manifolds, for a 2-user MIMO MAC with $T=4$, $M=2$, $N=10$ and $L=16$.}
	\label{fig:Manifolds_T4M2N10L16_beta_delta}
\end{figure}

\begin{figure}[H]
    \centering
\includegraphics[width=.6\textwidth]{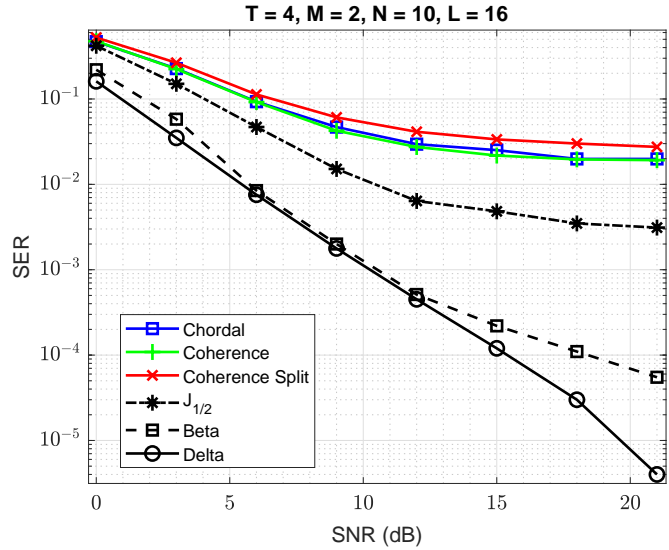}
     \caption{Performance of jointly optimized designs for a 2-user MIMO MAC with $T=4$, $M=2$, $N=10$ and $L=16$, vs. single-user designs.}
	\label{fig:Compare_T4M2N10L16_beta_delta}
\end{figure}


In Fig. \ref{fig:Manifolds_T6M3N3L16_beta_delta} and \ref{fig:Compare_T6M3N3L16_beta_delta} the setup changes to consider the same number of antennas at the transmitters and the receiver, $M=N=3$, with $T=6$ symbol periods. The same analysis of the previous scenario applies as well, with the single-user designs showing a noise floor at a SER value of around $10^{-1}$, whereas the delta function designs perform down to below $10^{-5}$ SER at $20$ dB. Nevertheless, we can see in Fig. \ref{fig:Manifolds_T6M3N3L16_beta_delta} that in this case the optimization on the oblique manifold performs very close to the trace manifold, suggesting that for this configuration the extra degrees of freedom by optimizing the average transmit power instead of using per-codeword power constraints does not significantly affect performance. The $J_{1/2}$ constellation seems to start developing a noise floor near $20$ dB, more than an order of magnitude below the single-user noise floor. But just like in the previous cases, the $\delta$ function designs consistently outperform all the other packings considered. 


\begin{figure}[H]
    \centering
\includegraphics[width=.6\textwidth]{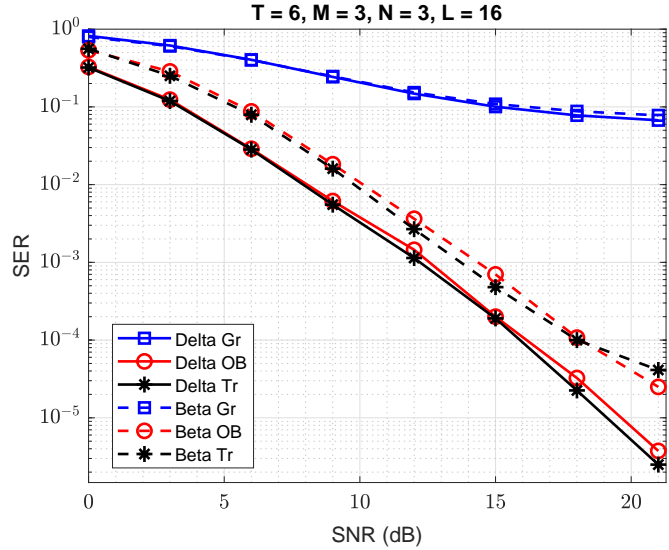}
     \caption{Performance of multi-user optimization designs in different Riemannian manifolds, for $T=6$, $M=3$, $N=3$ and $L=16$.}
	\label{fig:Manifolds_T6M3N3L16_beta_delta}
\end{figure}

\begin{figure}[H]
    \centering
\includegraphics[width=.6\textwidth]{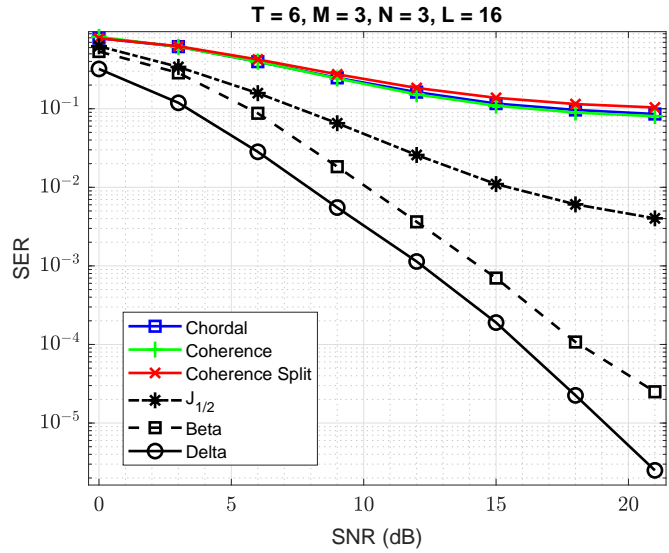}
     \caption{Performance of jointly optimized designs ($J_{1/2},\beta,\delta$) with $T=6$, $M=3$, $N=3$ and $L=16$, vs. single user designs.}
	\label{fig:Compare_T6M3N3L16_beta_delta}
\end{figure}

We can conclude that the impact of allowing less restrictive constraints on the codeword powers results in a better performance of the designed constellations when there is no full-diversity in the MAC. This is reasonable since the manifolds on which the codewords are represented, and moved during the optimization, have higher dimension the less constrained the power is, i.e. there is more space to approach possible (local) minima of the cost functions. However, this conclusion only holds for non-full diversity designs, as we shall see next in the full-diversity scenario.

\subsection{Noncoherent multiuser constellation designs in full-diversity scenarios}\label{subsec:resultsUB}

In this subsection we study the SER performance of the multiuser designs proposed in Sec. \ref{sec:ubpep} obtained by optimizing the union bound of the dominant term of the Brehler-Varanasi asymptotic PEP formula. We consider a 2-user MIMO MAC and work under the following assumptions: i) the two users have the same average SNR, and ii) there is no correlation between the channel fading coefficients.  Moreover, the formula of interest is only valid in the full-diversity case, meaning that only scenarios with $T\geq (K+1)M$ shall be analyzed here. Moreover, only the terms in the union bound corresponding to a single user in error are considered. There are two reasons for this simplification: first, the terms with only one user in error dominate the PEP expression; and second, this reduces the computational complexity dramatically, as explained in Sec. \ref{sec:ubpep}.

In Fig. \ref{fig:PEP_T3M1N3L16} the case of $T=3$ symbol periods, $M=1$ emitter antennas, $N=3$ receiver antennas, and $B=4$ bits per symbol is studied and compared versus the single-user designs. Two multiuser designs outperform these single-user constellations: the proposed union bound optimization of criterion (\ref{eq:MACUBfinal2}) and a min-max criterion (labeled MinMax-PEP), minimizing the worst PEP. One can understand that improving at every iteration the dominant term out of the possible pairwise probability errors ought to yield performance gains, which indeed is the case as shown by the dashed curve vs. the single user constellations. However, the union bound optimization clearly outperforms this by around $2.5$ dB at SER = $10^{-3}$. This is expected since a union bound method minimizes all terms of the possible error probabilities at the same time.

\begin{figure}
    \centering
\includegraphics[width=.6\textwidth]{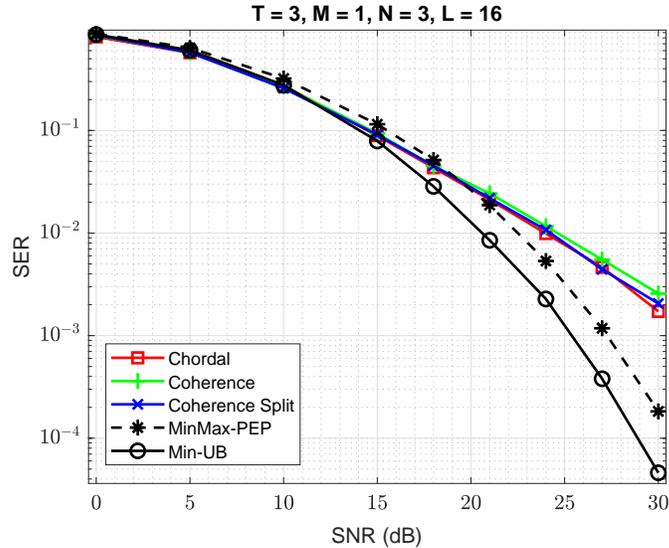}
     \caption{Performance of a jointly optimized design based on the Brehler-Varanasi asymptotic PEP union bound, with $T=3$, $M=1$, $N=3$ and $L=16$ vs. single user designs and compared to minimizing the dominant asymptotic PEP.}
	\label{fig:PEP_T3M1N3L16}
\end{figure}

A very similar configuration is shown in Fig. \ref{fig:PEP_T4M1N4L32}, where we consider $B=5$ bits per codeword, $N=4$ receive antennas, and $T=4$. In this case the gap between the min-max method and the union bound reduces, but the latter still provides the best results. It is interesting to note that the coherence criterion for single-user constellations outperforms the chordal distance criterion in the multiuser scenario, a behavior which was not so evident in the previous figure.

\begin{figure}[H]
    \centering
\includegraphics[width=.6\textwidth]{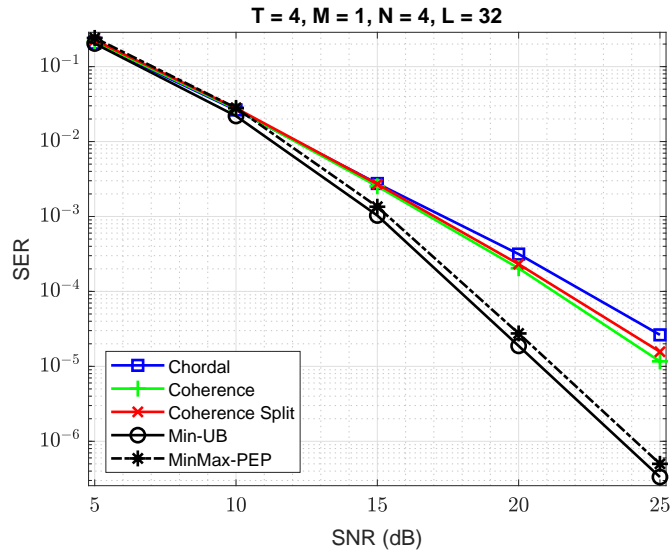}
     \caption{Performance of jointly optimized designs based on the PEP union bound, with $T=4$, $M=1$, $N=4$ and $L=32$, vs. single user designs.}
	\label{fig:PEP_T4M1N4L32}
\end{figure}

We  may conclude from the previous two figures that single-user codebooks do not only perform worse in term of SER at any given SNR, but also they do not achieve the same slope as the multiuser constellations. On the other hand, the multiuser codebooks designed with either the UB or a max-min approach attain the full-diversity of the system $MN$ for both users. In comparison to the min-max approach the UB criterion provides some coding gain, a shift to the left of the SER vs. SNR curve.


The impact of using different manifolds in the joint union bound criterion can be seen in Fig. \ref{fig:PEPman_T6M2N4}. Essentially there are not significant differences in performance. Moreover, for full-diversity scenarios the Grassmannian constellations seem to perform slightly better the higher the spectral efficiency is.

\begin{figure}[H]
    \centering
\includegraphics[width=.6\textwidth]{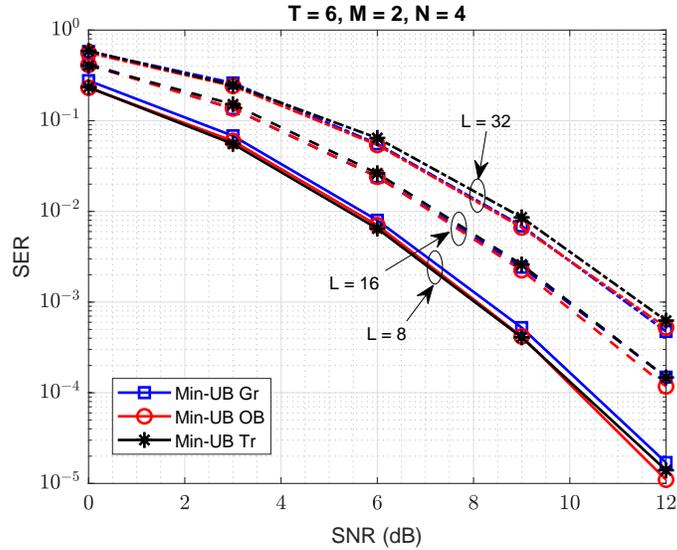}
     \caption{Comparison of the  multi-user constellation designs optimizing the Brehler-Varanasi union bound on different manifolds, for $T=6$, $M=2$, $N=4$ and $L \in \{8, 16, 32\}$.}
	\label{fig:PEPman_T6M2N4}
\end{figure}

It is important to point out that the performance of a given optimized constellation in the MAC depends on many parameters: number of users, number of antennas, coherence time, etc. In particular, the number of receive antennas affects the diversity (the slope of the SER vs. SNR curve) and hence can lead to important differences in the performance. This is illustrated in Fig. \ref{fig:PEP_T3M1N5}, which compares the performance of all unstructured joint constellation design methods proposed in this report. The full-diversity scenario is a 2-user MAC with $T=3$, $M=1$, $N=5$ and codebooks of cardinality $L=32$. It is remarkable that the behaviour of the union bound codebook is extremely good at high SNR, as expected from the theoretical result it rests upon, and much better than any other method, including the optimization of the min-max PEP. We can also compare in this plot the performance of the $\beta$ and $\delta$ designs, providing evidence that they are not good criteria for full-diversity scenarios. In fact, they do not even reach the single-user designs' performance. For full-diversity scenarios, the $\beta$ and $\delta$ designs appear to develop a noise floor at very high SNR, whereas the union bound criterion does not show any noise floor and attains the full diversity of the system $MN$. 

\begin{figure}[H]
    \centering
\includegraphics[width=.6\textwidth]{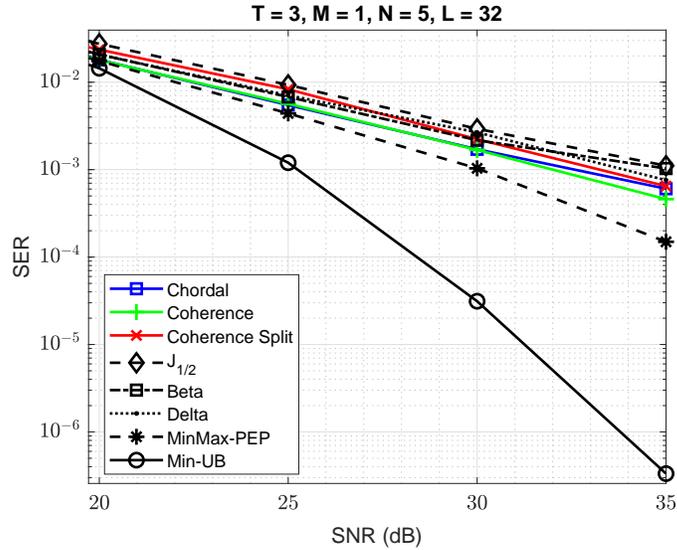}
     \caption{Comparison of all the different constellation designs based on the proposed multi-user cost functions, for $T = 6$, $M = 2$, $N = 5$ and $L = 32$, vs. the single user designs.}
	\label{fig:PEP_T3M1N5}
\end{figure}

Finally, in Fig. \ref{fig:PEP_T6M2two} we compare the multiuser design criteria in a full-diversity scenario with $T=6$, $M=2$ and different number of receive antennas and bit rates. It is worth mentioning that the gap in performance between the Brehler-Varanasi asymptotic PEP and the Ngo-Yang proxy functions seems to get reduced when increasing the number of bits per codeword. Still, the former outperforms all designs studied so far. Moreover, since $J_{1/2}$ is the best performing criteria out of five shown in \cite{ngoyang}, we can conclude that our designs provide state-of-the-art multiuser constellations for the MAC.


\begin{figure}[H]
    \centering
\includegraphics[width=.6\textwidth]{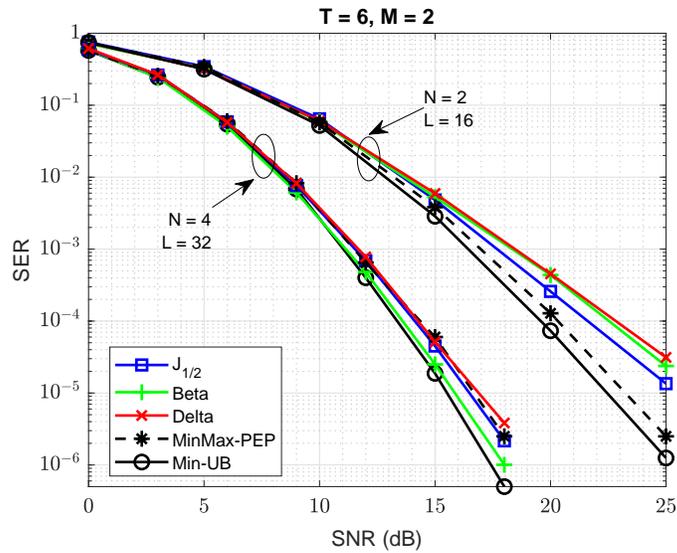}
     \caption{Comparison of the different multi-user constellation designs based on the proposed multi-user cost functions, for $T = 6$, $M = 2$, $N \in \{2,4 \}$ and $L \in \{16,32\}$.}
	\label{fig:PEP_T6M2two}
\end{figure}
\section{Conclusion} \label{sec:conclusions}

In this paper we have developed Riemmanian optimization techniques for designing noncoherent constellations for the MIMO MAC. In particular, we have developed optimized multiuser space-time codebooks for full-diversity ($T \geq (K+1)M$) and non-full diversity scenarios ($T < (K+1)M$). For full-diversity scenarios, the cost function is a union bound of the dominant terms (i.e, those terms corresponding to the case where only one of the users of the MAC channel is in error) of the asymptotic PEP. For non-full diversity scenarios the PEP expression is no longer valid and therefore we use union bounds of some recently proposed proxies pf the PEP, called the $\delta$ and $\beta$ functions, as design criteria. The proposed cost functions and the corresponding Riemannian optimization techniques are valid for any number of users. 

In addition to the traditional Grassmann manifold, which is optimal only in the single-user case, we consider the optimization of multiuser codebooks in other Riemannian manifolds corresponding to different power constraints on the codewords. We show that the manifold on which the optimization is performed can have a significant impact on performance, especially in non-full diversity scenarios. Our results suggest that in non-full diversity case the $\delta_{UB}$ cost function optimized on the trace manifold, corresponding to an average power constraint, is the best performing design. Whereas in the full-diversity case the best performing constellations in terms of symbol error rate (SER) are those designed using the dominant factor of the asymptotic joint PEP on the Grassmann manifold. Future lines of work include the development of an asymptotic PEP formula for non-full diversity scenarios that would avoid the use of proxies in this case and the study of noncoherent schemes for the broadcast channel.



\appendix
\subsection{Riemannian manifolds}\label{app:preliminaries}
The complex Grassmannian $\Gras$ is the set of $M$--dimensional complex subspaces of $\mathbb{C}^T$, with $T>M$, that is a complex manifold of dimension $M(T-M)$. Elements in $\Gras$ are represented by matrices in the Stiefel manifold $\A\in\St$, that is $\A\in\mathbb{C}^{T\times M}$, $\A^H\A=\I_M$. This representation is not unique, since $\A$ and $\A\U$ with $\U$ a unitary $M\times M$ matrix represent the same element in $\Gras$, so formally we should denote elements of the Grassmannian as $[\A]$ where $\A\in\St$ is a unitary basis for that subspace, ${\bf P}_{\bf A} = {\bf A}{\bf A}^H$ denotes the orthogonal projection onto $[{\bf A} ]$ and $[\A]$ is the class of $\A$ under the quotient by the set of $M\times M$ unitary matrices $\mathcal{U}_M$. Mathematically, this defines a Riemannian structure on the Grassmannian given by the {\em Riemannian submersion}
\[
\begin{matrix}
\pi:&\St&\to&\Gras=\St/\mathcal{U}_M\\
&\A&\mapsto&[\A].
\end{matrix}
\]

Sometimes we will consider the optimization of a real function $\varphi$ whose argument can be either a complex matrix in the ambient space $\X \in \mathbb{C}^{T \times M}$, a Stiefel matrix $\X \in \St$, or a point in the Grassmanian $[\X] \in \Gras$. We will denote the function generically as $\varphi(\X)$, meaning for the Grassmann that $\varphi(\X\U) = \varphi(\X)$ for any $M\times M$ unitary matrix $\U$, and employ the notation $D \varphi(\X)$ to denote the unconstrained derivative of the function in the ambient space, and $\nabla \varphi(\X)$ to denote the gradient of the function on the tangent space of the Grassmannian. In both cases it will be understood that the derivative or the gradient is evaluated at $\X$ or $[\X]$, respectively. In particular, these derivatives play two roles: on the one hand, the unconstrained derivative $D \varphi(\X)$, depending only on a point $\X$, is the Jacobian matrix of $\varphi$ with respect to the components $X_{mn}$ of $\X$, for $m=1,\dots,T,\; n=1,\dots, M$, i.e., as a matrix it has complex components given by
$$
D \varphi(\X)_{mn} = \frac{\partial\varphi}{\partial X_{mn}} = \frac{\partial\varphi}{\partial \Re(X_{mn})}+ i\,\frac{\partial\varphi}{\partial \mathfrak{I}( X_{mn})}.
$$
On the other hand, these derivatives, when depending both on a point $\A$ and a tangent vector $\dot\A$, are to be understood as directional derivatives in their respective tangent spaces, for example
$$
D \varphi(\A)(\dot\A) = \lim_{t\rightarrow 0}\frac{\varphi(\A + t\dot\A)-\varphi(\A)}{t}=\frac{d}{dt}\mid_{t=0}\varphi(\A + t\dot\A)
$$
With this interpretation we can define partial derivatives of $\varphi$ with respect to the real and imaginary part of every direction in the tangent space and thus arrive at the Jacobian matrix again. The relationship between both objects in the ambient space can be verified to be:
\begin{equation}\label{eq:gradient_def}
D \varphi(\A)(\dot\A) = \Re(\langle D\varphi(\A) , \dot \A\rangle_F),
\end{equation}
a property which will serve as requirement for the definition of gradient vector $\nabla\varphi$ on a general manifold (see Corollary \ref{cor:gradient}).

The tangent space to the Stiefel manifold at $\A\in\St$ is easy to describe from the defining equation $\A^H\A=\I_M$:
\begin{align*}
T_{\A}\St=&\left\{\dot \A\in\mathbb{C}^{T\times M}:\frac{d}{dt}\mid_{t=0}((\A+t\dot \A)^H(\A+t\dot \A))= 0\right\}\\
=&\{\dot \A\in\mathbb{C}^{T\times M}:\dot \A^H\A+\A^H\dot \A={\bf 0}\}.
\end{align*}
The Riemannian submersion $\pi$ allows us to identify the tangent space to the Grassmannian with the orthogonal to the kernel of $\pi$; in other words,
\begin{align*}
T_{[\A]}\Gras\equiv&\{\dot \A\in T_{\A}\St:\dot \A\perp \A\dot \U \text{ for all } \dot \U\in T_{\I_M}\mathcal{U}_M\}\\
=&\{\dot \A\in\mathbb{C}^{T\times M}:\dot \A^H\A+\A^H\dot \A={\bf 0}, \langle \dot \A,\A\dot \U\rangle_F=0\;\forall \, \dot \U:\dot \U+\dot \U^H={\bf 0}\}\\
=&\{(\I_T-\A\A^H)\dot \B:\dot \B\in \mathbb{C}^{T\times M}\}.
\end{align*}
For this note that both spaces have the same (complex) dimension $M(T-M)$ and that the latter is included in the former since for $\dot \B\in \mathbb{C}^{T\times M}$, taking $\dot \A=(\I_T-\A\A^H)\dot \B$, we have:
\begin{align*}
\dot \A^H\A+\A^H\dot \A=&\dot \B^H(\I_T-\A\A^H)\A+\A^H(\I_T-\A\A^H)\dot \B={\bf 0},\\
\langle \dot \A,\A\dot \U\rangle_F=&\langle (\I_T-\A\A^H)\dot \B,\A\dot \U\rangle_F=\langle \A^H(\I_T-\A\A^H)\dot \B,\dot \U\rangle_F=0.
\end{align*}
We thus obtain
\[
T_{[\A]}\Gras\equiv\{(\I_T-\A\A^H)\dot \B:\dot \B\in \mathbb{C}^{T\times M}\}.
\]
The last set obviously does not depend on the chosen representative for $[\A]$.

The following lemma is fundamental for the computation of gradients in the Grassmannian that will be used in the UB optimization algorithm.

\begin{corollary}\label{cor:gradient}
Let $\varphi:\mathbb{C}^{T\times M}\to\mathbb{R}$ be a $C^1$ mapping, defined at least in some open neighborhood of the Stiefel manifold $\St\subseteq\mathbb{C}^{T\times M}$, and assume that $\varphi$ can be defined as a function on $\Gras$, that is, we have:
$$
\varphi(\A)=\varphi(\A\U)\text{ for $\A\in\St$, $\U\in \mathcal U_M$}.
$$
Then, the gradient of $\varphi$ at $\A\in\St$ as a Grassmannian mapping is:
$$
\nabla\varphi(\A) = (\I_T-\A\A^H) D \varphi(\A),
$$
where $D \varphi$ is the unconstrained gradient of $\varphi$ as a function on the ambient space $\mathbb{C}^{T\times M}$.
\end{corollary}
\begin{proof}
  By definition, the gradient $ \nabla \varphi(\A)$ is the unique element of $T_{[\A]}\Gras$ such that for all $\dot \A\in T_{[\A]}\Gras$:
  \begin{align*}
  D\varphi(\A)(\dot \A)=\Re\langle \nabla \varphi(\A),\dot \A\rangle.
  \end{align*}
  Let $\dot \A=(\I_T-\A\A^H)\dot \B\in T_{[\A]}\Gras$ and note that
  \begin{align*}
&  \Re\langle (\I_T-\A\A^H) D \varphi(\A),\dot \A\rangle=\Re\langle (\I_T-\A\A^H) D \varphi(\A),(\I_T-\A\A^H)\dot \B\rangle\\
 & =\Re\langle D \varphi(\A),(\I_T-\A\A^H)\dot \B\rangle
  =\Re\langle D \varphi(\A),\dot \A\rangle
  = D\varphi(\A) (\dot \A),
  \end{align*}
  and since $(\I_T-\A\A^H) D \varphi(\A)$ is an element of $T_{[\A]}\Gras$, it satisfies the definition of gradient.
\end{proof}

\section*{Acknowledgments}
This work was supported by Huawei Technologies, Sweden under the project GRASSCOM. The work of D. Cuevas was also partly supported under grant FPU20/03563 funded by Ministerio de Universidades (MIU), Spain. The work of Carlos Beltr{\'a}n was also partly supported under grant PID2020-113887GB-I00 funded by MCIN/ AEI /10.13039/501100011033. The work of I. Santamaria was also partly supported under grant PID2019-104958RB-C43 (ADELE) funded by MCIN/ AEI /10.13039/501100011033.

\bibliographystyle{ieeetr}
\bibliography{IEEE_IT}

\vfill

\end{document}